%% file: main.tex
\documentclass{article}

\PassOptionsToPackage{numbers, compress}{natbib}
 \usepackage[main,final]{neurips_2026}

\usepackage[utf8]{inputenc} 
\usepackage[T1]{fontenc}    
\usepackage{hyperref}       
\usepackage{url}            
\usepackage{booktabs}       
\usepackage{amsfonts}       
\usepackage{nicefrac}       
\usepackage{microtype}      
\usepackage{xcolor}         
\usepackage[normalem]{ulem}
\NewDocumentCommand{\camera}{t2 +m}{%
    {\IfBooleanTF{#1}{\color{black}}{\color{black}}#2}%
}
\DeclareRobustCommand{\removed}[1]{\sout{#1}}

\usepackage{enumerate}
\usepackage{enumitem}
\usepackage{caption} 
\usepackage{wrapfig}

\usepackage{microtype}
\usepackage{graphicx}
\usepackage{subcaption}
\usepackage{booktabs}

\usepackage{hyperref}

\usepackage{amsmath,amssymb}
\usepackage{algorithmic}  
\newcounter{algctr}
\renewcommand{\thealgctr}{\arabic{algctr}}
\newcommand{\algcaption}[1]{%
  \refstepcounter{algctr}%
  \noindent\textbf{Algorithm \thealgctr}: #1%
}

\usepackage{amsmath}
\usepackage{amssymb}
\usepackage{mathtools}
\usepackage{amsthm}
\usepackage{amsfonts}       

\usepackage[capitalize,noabbrev]{cleveref}

\theoremstyle{plain}
\newtheorem{theorem}{Theorem}[section]
\newtheorem{proposition}[theorem]{Proposition}

\theoremstyle{definition}

\theoremstyle{remark}

\usepackage[textsize=tiny]{todonotes}

\title{Pushing Biomolecular Utility-Diversity Frontiers with Supergroup Relative Policy Optimization}

\author{%
  \textbf{Xinwu Ye}$^{1,2,3}$\thanks{Equal contribution. Email: \texttt{xinwuye43@connect.hku.hk}}
  \quad
  \textbf{He Cao}$^{3}$\footnotemark[1]
  \quad
  \textbf{Hao Li}$^{4}$
  \quad
  \textbf{Bin Feng}$^{2}$ \\
  \textbf{Zijing Liu}$^{2}$
  \quad
  \textbf{Xiangru Tang}$^{5}$
  \quad
  \textbf{Yu Li}$^{2}$\thanks{Corresponding authors.}
  \quad
  \textbf{Shenghua Gao}$^{1}$\footnotemark[2] \\
  \\[-1mm]
  $^{1}$The University of Hong Kong \\
  $^{2}$International Digital Economy Academy \\
  $^{3}$AbIntelli Bio \\
  $^{4}$Peking University \\
  $^{5}$Yale University
}

\begin{document}

\maketitle

\begin{abstract}
Biomolecular generators are often adapted with reward feedback to improve task-specific utility, but pushing utility alone can concentrate generation on a narrow family of candidates. Maintaining diversity is difficult because sample diversity is a set-level property. We introduce Supergroup Relative Policy Optimization (SGRPO), a flexible GRPO-style framework that directly constructs rewards from set-level diversity. For each condition, SGRPO samples a supergroup of candidate sets, compares their diversity under the same condition, and redistributes the group diversity reward to individual rollouts through leave-one-out diversity contributions before combining it with rollout-level utility. This design decouples SGRPO from a particular generator, utility reward, or diversity metric, and allows instantiation with different GRPO-style approaches. We evaluate SGRPO on \camera2{three main settings:} \emph{de novo} small-molecule design, pocket-based small-molecule design, and \emph{de novo} protein design, instantiating it with both GRPO and Coupled-GRPO across autoregressive and discrete diffusion generators. Across decoding sweeps, SGRPO expands the utility-diversity Pareto frontier and achieves the best frontier-level metrics relative to baselines. Our analyses further show that direct set-level diversity rewards remain effective with small groups and help preserve broader generation-distribution coverage during post-training.
The code is available at \url{https://github.com/IDEA-XL/SGRPO}.
\end{abstract}

\input{sec/intro}

\input{sec/related}

\input{sec/setup}

\input{sec/method}

\input{sec/exp}

\input{sec/analysis}

\input{sec/conclusion}

\begin{ack}
\camera2{This work is partially supported by funding and computing resource provided by HKU-Shanghai Intelligent Computing Research Center. This work was also supported by the General Research Fund of
the Research Grants Council of Hong Kong under Grant No. 17200725 and the JC STEM Lab funded by The Hong Kong Jockey Club Charities Trust.}

\paragraph{Competing interests.}
\camera2{The authors declare no competing interests.}
\end{ack}

\newpage
\bibliographystyle{plainnat}
\bibliography{ref}

\newpage
\appendix
\onecolumn

\input{sec/appendix}

\newpage
\input{checklist.tex}

\end{document}

%% file: sec/intro.tex
\section{Introduction}

Biomolecular generation aims to produce candidates that satisfy chemical or biological design objectives, and reinforcement learning (RL) provides a natural framework for post-training pretrained generators from reward feedback toward desired properties, structures, or functions~\citep{olivecrona2017molecular}. In practice, however, generation quality is not determined by utility alone. A model that maximizes a property score, docking proxy, or protein-level objective may concentrate probability mass on a narrow family of candidates, while a highly diverse generator may fail to deliver enough high-utility samples. This creates a utility-diversity trade-off: different downstream settings may prefer different operating points, often modulated by decoding choices such as temperature, so the relevant objective is not a single best reward value but an improved Pareto frontier of attainable utility-diversity pairs. While many successful RL approaches are closely tailored to specific model classes, molecular or protein representations, and design settings~\citep{olivecrona2017molecular,you2018graph,ektefaie2024reinforcement,wang2025proteinzero}, our goal is a broadly applicable post-training principle. We evaluate it across different generator families, conditioning settings, utility functions, and diversity metrics, while leaving broader task-specific instantiations to future work.

A more broadly applicable class of diversity-aware RL methods encourages diversity through memory- or history-dependent novelty penalties~\citep{blaschke2020memory,loeffler2024reinvent}. Such methods down-weight candidates that are too similar to previously sampled molecules, scaffolds, clusters, or neighborhoods, and can be effective in practice. However, novelty relative to past samples is only an indirect surrogate for the diversity of the current candidate set produced under a given condition. As a result, these methods may over-penalize useful high-density modes or induce distributional drift during post-training. More fundamentally, the target quantity itself is set-level: diversity is defined over collections of samples, whereas policy optimization updates individual rollouts. This raises the central question of this paper: \emph{can we optimize sample-set diversity directly, as a first-class objective, while still assigning useful credit to individual generated candidates?}

We address this with Supergroup Relative Policy Optimization (SGRPO), a simple framework for directly \camera2{rewarding} sample-set diversity together with rollout-level utility. For each condition, SGRPO samples multiple candidate sets from the current policy, scores each set using a user-specified diversity metric, and compares sets only against other sets generated under the same condition. To make this set-level signal actionable for policy learning, SGRPO redistributes each set’s diversity reward to its members through leave-one-out diversity contributions, so candidates that genuinely support set diversity receive stronger credit. The resulting \emph{supergroup}-relative advantage can be instantiated with different GRPO-style optimizers.

We instantiate SGRPO with two GRPO-style optimizers and evaluate it on three \camera2{main} biomolecular generation settings: unconditional \emph{de novo} small-molecule design with GenMol~\citep{lee2025genmol}, pocket-based small-molecule design with GenMol-P, our pocket-conditioned variant of GenMol, and unconditional \emph{de novo} protein design with ProGen2~\citep{nijkamp2023progen2}. 
Across decoding sweeps, SGRPO consistently improves the attainable utility-diversity Pareto frontier over baselines. 
It remains effective even with small group sizes and better preserves generation-distribution coverage during post-training, showing that directly \camera2{rewarding} set-level diversity can yield robust gains across both molecule and protein generation.

%% file: sec/related.tex
\section{Related Work}

\subsection{Objective Optimization in Biomolecular Generation}
\label{subsec:objective_optimization}

Objective optimization in biomolecular generation is commonly approached either by conditioning or guiding generators toward desired properties, structures, or functions~\citep{lim2018molecular,kotsias2020direct,bagal2021molgpt,jolicoeur2024any,dauparas2022robust,runcie2023silvr,xiong2025proteinguide,Tang2024.10.28.620755}, or by improving candidates from oracle feedback through methods such as latent-space Bayesian or evolutionary optimization, iterative retraining, preference optimization, and reinforcement learning~\citep{gomez2018automatic,griffiths2020constrained,castro2022transformer,brookes2018design,tripp2020sample,ye2026latentchem,cheng2024decomposed,widatalla2024aligning,li2026cagenmol}. We focus on the RL branch, which provides a general feedback-driven post-training formulation in which generated candidates are scored by objectives such as molecular properties, stability, or multi-objective reward functions, and the generator is updated to increase the likelihood of high-reward samples. RL-based biomolecular optimization has been instantiated across diverse representations, including SMILES sequence models such as REINVENT, ReLeaSE, and ChemRLformer~\citep{olivecrona2017molecular,popova2018deep,ghugare2023searching}, graph- or fragment-based molecular generators such as GCPN, MolDQN, RationaleRL, LibINVENT, and DrugEx v3~\citep{you2018graph,zhou2019optimization,jin2020multi,fialkova2021libinvent,liu2023drugex}, and protein sequence or structure-conditioned generators such as model-based RL for biological sequence design, RL-DIF, and ProteinZero~\citep{angermueller2019model,ektefaie2024reinforcement,wang2025proteinzero}. This broad applicability of RL motivates our focus on diversity-aware reward design at the post-training level.

\subsection{Diversity-Promoting RL for Biomolecular Generation}
\label{subsec:diversity_promoting_rl}

Diversity-promoting RL has been explored in both molecular and protein generation to mitigate mode collapse and sample redundancy. One line of work builds diversity objectives around the structure of a specific generator or design task, for example by jointly generating multiple SMILES strings in a single sequence~\citep{jang2024can}, exploiting augmented SMILES and score reuse~\citep{bjerrum2023faster}, incorporating diversity into fragment-based molecular construction~\citep{yang2021hit}, pairing exploitation and exploration policies during generation~\citep{liu2019exploration}, or adding task-specific regularization in protein inverse folding and sequence design~\citep{ektefaie2024rldif,wang2025proteinzero,park2024improving}. A more broadly applicable family instead promotes diversity through indirect reward shaping, such as diverse mini-batch selection~\citep{svensson2025diverse}, memory- or scaffold-based penalties and filters~\citep{blaschke2020memory,pereira2021diversity,loeffler2024reinvent,thomas2022augmented,gummesson2024utilizing,zhu2025scaffold}, distance-to-memory or novelty rewards~\citep{hu2024hamiltonian,svensson2024diversity,park2025mol,chadi2023curiosity}, entropy regularization~\citep{seqdiffql2025}, or count-based visitation bonuses~\citep{angermueller2020model}. These approaches have shown empirical benefits, but they are either tightly coupled to particular generator interfaces or optimize indirect proxies such as novelty, entropy, or history-relative exploration rather than the diversity of the current generated sample set itself. Our work focuses on this latter gap.

%% file: sec/setup.tex
\section{Problem Setup: Utility--Diversity Frontier in Biomolecular Generation}
\label{sec:problem_setup}

\subsection{Setup}
\label{subsec:preliminaries}

We consider a conditional biomolecular generator $\pi_\theta(x \mid \mathcal{C})$, where $x$ is a generated candidate and $\mathcal{C}$ is the conditioning input. Depending on the task, $\mathcal{C}$ may be empty, a task or property specification, or a target environment such as a protein binding pocket. The formulation is model-agnostic and applies to pretrained \emph{de novo} molecular generators, pocket-conditioned molecular generators, and protein language models.
Each candidate receives an individual utility score $r(x,\mathcal{C})$. The exact form of $r$ depends on the domain. For small molecules, it may combine drug-likeness and synthesizability, and in pocket-conditioned generation, it may additionally include target-specific terms such as docking. For proteins, utility may reflect sequence plausibility, stability, foldability, or developability.

We also care about diversity among the generated outputs. For a set of $K$ candidates generated under the same condition, denoted by $G = \{x_1,\dots,x_K\}$, let $D(G)$ be a set-level diversity score. This score may measure internal diversity, scaffold diversity, sequence diversity, or cluster coverage. The key point is that diversity is \emph{not} a per-sample reward: in general, it depends on the relationships among samples in the set and cannot be reduced to independently scoring each candidate. Optimizing diversity, therefore, requires reasoning over groups of outputs rather than isolated generations.

\subsection{Utility--diversity frontier}
\label{subsec:pareto_objective}

Let $p(\mathcal{C})$ denote the distribution over generation conditions. At inference time, the trained generator is paired with a decoding strategy $a \in \mathcal{A}$, such as a sampling temperature or related decoding hyperparameters. Together, $(\theta, a)$ determine two expected quantities: the expected individual utility, denoted by $U(\theta,a)$, and the expected set-level diversity, denoted by $V(\theta,a)$. Here $U(\theta,a)$ is computed from single generated samples, while $V(\theta,a)$ is computed from sets of $K$ samples drawn under the same condition.

Varying the decoding strategy induces a set of attainable utility--diversity trade-offs for the generator, which we denote by
$
\mathcal{P}(\theta)
=
\{(U(\theta,a), V(\theta,a)) : a \in \mathcal{A}\}.
$
A point on this set is Pareto-optimal if no other decoding strategy achieves both higher utility and higher diversity at the same time.

Our goal is to improve this frontier itself. Rather than optimizing only utility or only diversity, we seek post-training methods that push $\mathcal{P}(\theta)$ outward, so that the same generator can achieve better utility at a fixed diversity level, better diversity at a fixed utility level, or both.

%% file: sec/method.tex
\section{Supergroup Relative Policy Optimization}
\label{sec:sgrpo}

SGRPO is a post-training reinforcement learning method for improving the utility--diversity frontier of a pretrained biomolecular generator. Its central idea is simple: since diversity is a set-level property, training should compare \emph{sets} of candidates generated under the same condition, rather than scoring each candidate in isolation. For each condition $\mathcal{C}$, SGRPO samples several candidate groups, scores each group by diversity, redistributes the group-level signal back to individual candidates according to their within-group contribution, and then applies a PPO-style update using a same-condition relative advantage. Figure~\ref{fig:sgrpo_overview} illustrates the overall pipeline, and detailed pseudocode is provided in Appendix~\ref{app:full_algorithm}.

\begin{figure*}[!ht]
    \centering
    \includegraphics[width=\textwidth]{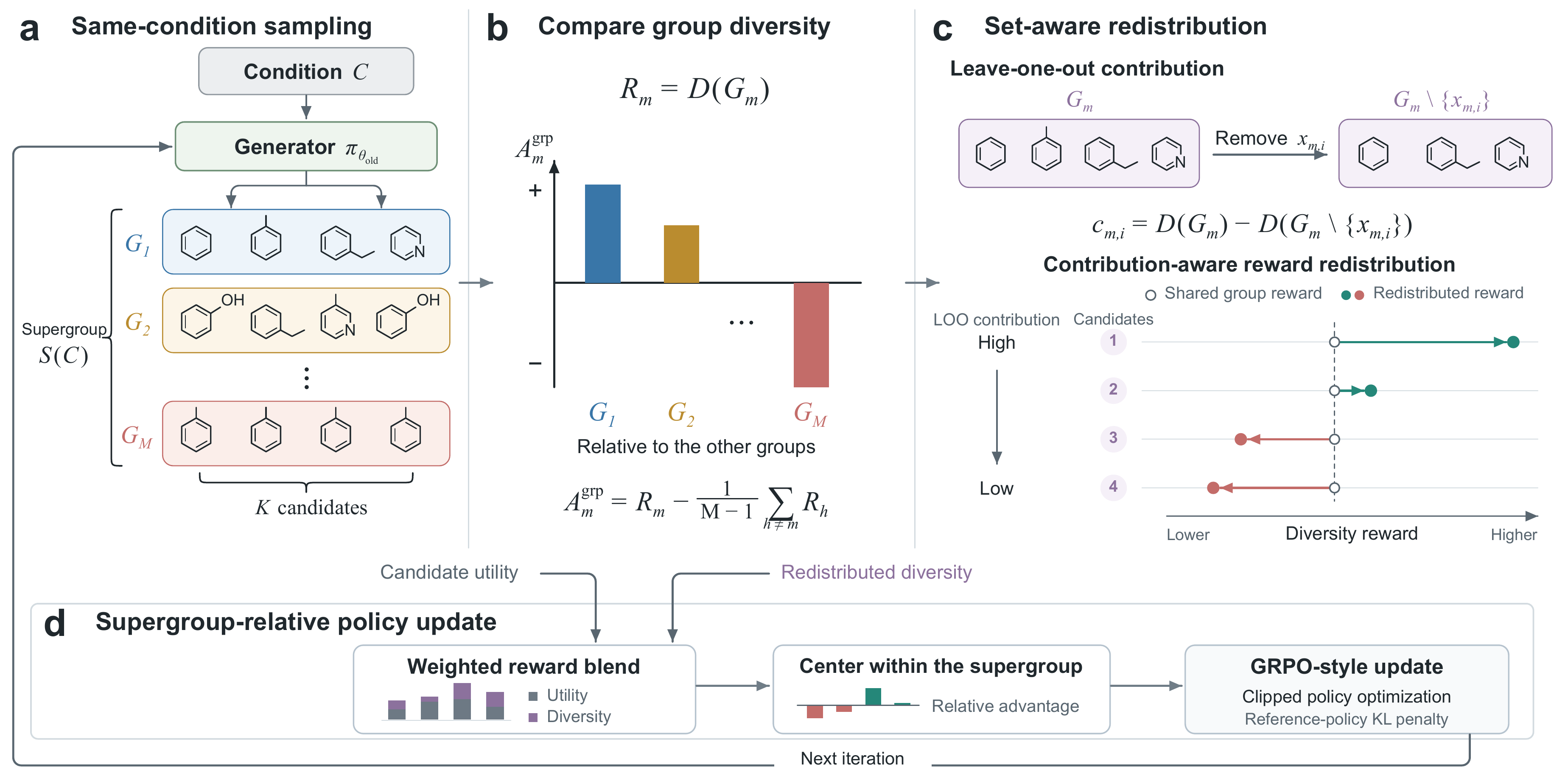}
    \caption{
        Overview of SGRPO. For each condition, SGRPO samples a same-condition supergroup, computes rollout-level utility and group-level diversity, compares groups by leave-one-out group-relative diversity, redistributes the diversity signal within each group according to leave-one-out set contributions, centers the composed rewards over the supergroup, and updates the policy with a PPO-style objective and KL regularization.
    }
    \label{fig:sgrpo_overview}
\end{figure*}

\subsection{Same-condition supergroups}
\label{subsec:supergroup_sampling}

For a condition $\mathcal{C}$, we sample $M$ groups from the old policy, each containing $K$ independently generated rollouts. Denote the resulting collection by $\mathcal{S}(\mathcal{C})=\{G_1,\dots,G_M\}$, where $G_m=\{x_{m,1},\dots,x_{m,K}\}$ and each $x_{m,i}\sim \pi_{\theta_{\rm old}}(\cdot\mid\mathcal{C})$. We refer to $\mathcal{S}(\mathcal{C})$ as a \emph{supergroup}. It contains $N=MK$ candidates generated under the same condition, with $M$ controlling how many alternative groups are compared and $K$ controlling the size of each group.

Restricting comparisons to a single supergroup is important. In conditional biomolecular generation, different conditions can have very different intrinsic difficulty, so comparing samples across conditions would confound policy quality with condition difficulty. SGRPO instead performs only \emph{local} comparisons: under the same $\mathcal{C}$, which groups are more diverse, and which rollouts are more useful?

\subsection{Utility and group-level diversity}
\label{subsec:rollout_group_rewards}

Each rollout $x_{m,i}$ receives an individual utility reward $r_{m,i}=r(x_{m,i},\mathcal{C})$, while each group $G_m$ receives a diversity score $R_m=D(G_m)$. Thus, utility is defined at the candidate level, whereas diversity is defined over the whole group.

To compare groups generated under the same condition, we center group diversity within the supergroup. Let $\bar R=\frac{1}{M}\sum_{h=1}^{M}R_h$. We define the group-relative diversity signal as
\begin{equation}
A^{\mathrm{grp}}_m
=
R_m-\frac{1}{M-1}\sum_{h\neq m}R_h
=
\frac{M}{M-1}(R_m-\bar R).
\label{eq:group_advantage_centered}
\end{equation}
This is simply a leave-one-out comparison among same-condition groups: $A^{\mathrm{grp}}_m>0$ means that $G_m$ is more diverse than its alternatives, and $A^{\mathrm{grp}}_m<0$ means the opposite. For the normalized pairwise diversity used in our experiments, average diversity over groups of size $K$ is an unbiased proxy for the diversity of a larger same-condition sample, so optimizing group diversity is aligned with improving diversity at the supergroup level. Formal statements are given in Appendix~\ref{app:sgrpo_theory}.

\subsection{Set-aware redistribution}
\label{subsec:set_aware_redistribution}

The diversity score $R_m$ evaluates an entire group, but policy optimization ultimately acts on individual rollouts. SGRPO bridges this gap by assigning more of the diversity signal to candidates that matter more for the diversity of their own group.

For each rollout $x_{m,i}\in G_m$, we first compute its leave-one-out contribution
$c_{m,i}=D(G_m)-D(G_m\setminus\{x_{m,i}\})$,
and standardize these contributions within the group as
$z_{m,i}=(c_{m,i}-\bar c_m)/(\sigma(c_{m,\cdot})+\zeta)$.
We then form two sign-aware softmax weight vectors:
\begin{equation}
w^{\pm}_{m,i}
=
K\cdot
\frac{\exp(\pm z_{m,i}/\tau_c)}
{\sum_{j=1}^{K}\exp(\pm z_{m,j}/\tau_c)}.
\label{eq:redistribution_weights}
\end{equation}
By construction, $\sum_i w^+_{m,i}=\sum_i w^-_{m,i}=K$. Here $w^+_{m,i}$ emphasizes candidates with larger diversity contributions, while $w^-_{m,i}$ emphasizes candidates with smaller ones.

We then define the redistributed diversity reward
\begin{equation}
\widetilde R_{m,i}
=
R_m
+
[A^{\mathrm{grp}}_m]_+(w^+_{m,i}-1)
-
[-A^{\mathrm{grp}}_m]_+(w^-_{m,i}-1),
\label{eq:redistributed_group_reward}
\end{equation}
where $[a]_+=\max(a,0)$. \camera{Proposition~\ref{prop:reward_conservation} in Appendix~\ref{app:redistribution_objective} proves that this redistribution preserves the mean group reward.} The redistribution is sign-aware. If $A^{\mathrm{grp}}_m>0$, then $G_m$ is more diverse than its same-condition alternatives, and the extra positive signal is concentrated on candidates that contributed more to that diversity. If $A^{\mathrm{grp}}_m<0$, then $G_m$ is less diverse, and the negative signal is concentrated on candidates that contributed less. By construction, redistribution preserves the original group reward on average, i.e., $\frac{1}{K}\sum_{i=1}^{K}\widetilde R_{m,i}=R_m$.

\subsection{Supergroup-relative policy update}
\label{subsec:supergroup_advantage}

We combine candidate-level utility and redistributed diversity into a single reward,
$\widehat r_{m,i}=(1-\lambda)r_{m,i}+\lambda \widetilde R_{m,i}$,
where $\lambda\in[0,1]$ controls the utility--diversity trade-off. Let $\bar r_{\mathcal S}=\frac{1}{MK}\sum_{m=1}^{M}\sum_{i=1}^{K}\widehat r_{m,i}$ denote the average composed reward within the supergroup. The final supergroup-relative advantage is
\begin{equation}
A_{m,i}
=
\frac{MK}{MK-1}\bigl(\widehat r_{m,i}-\bar r_{\mathcal S}\bigr).
\label{eq:supergroup_advantage}
\end{equation}
Equivalently, this is a leave-one-out baseline over all rollouts in the same supergroup. Since all rollouts in the supergroup share the same condition, $A_{m,i}$ measures whether a rollout is better or worse than its local same-condition alternatives after utility and diversity have been combined.

We then update the policy with a clipped PPO objective and a KL penalty to a reference policy $\pi_{\rm ref}$. Let $\rho_{m,i}=\pi_\theta(x_{m,i}\mid\mathcal{C})/\pi_{\theta_{\rm old}}(x_{m,i}\mid\mathcal{C})$. The objective is
\begin{equation}
\mathcal{L}_{\mathrm{SGRPO}}(\theta)
=
-\mathbb{E}\!\left[
\min\!\Bigl(
\rho_{m,i}A_{m,i},
\operatorname{clip}(\rho_{m,i},1-\epsilon,1+\epsilon)A_{m,i}
\Bigr)
\right]
+
\beta\,\mathbb{E}\!\left[
\mathrm{KL}\!\Bigl(
\pi_\theta(\cdot\mid\mathcal{C})
\,\|\,
\pi_{\rm ref}(\cdot\mid\mathcal{C})
\Bigr)
\right]
\label{eq:sgrpo_objective}
\end{equation}
The expectation is over conditions, sampled supergroups, and rollouts. In practice, SGRPO alternates between sampling same-condition supergroups, computing group-level diversity and rollout-level redistributed rewards, and updating the policy with the objective above. Full pseudocode is provided in Appendix~\ref{app:full_algorithm}.

%% file: sec/exp.tex
\section{Experiments}
\newcommand{\metricci}[2]{\( #1{\scriptstyle \pm #2}\)}
\newcommand{\bestmetricci}[2]{\(\mathbf{#1{\scriptstyle \pm #2}}\)}

We evaluate whether SGRPO expands the utility-diversity Pareto frontier across three biomolecular generation settings: unconditional \emph{de novo} small-molecule design, pocket-based small-molecule design, and \emph{de novo} protein design. In each setting, we decode each model under a sweep of operating points and summarize every operating point by its utility and set-level diversity. We compare the resulting frontiers against \camera{the pretrained generator, GRPO, memory-assisted GRPO when applicable, Diverse Mini-Batch GRPO (DMB-GRPO), and Entropy-Regularized GRPO (ER-GRPO)}, using the same Pareto-level metrics across tasks. This evaluation tests whether supergroup-relative diversity pressure improves the trade-off frontier itself, rather than merely shifting generation toward higher utility or higher randomness.

\subsection{Evaluation Protocol}

Each experiment evaluates a generator under a range of task-specific decoding settings, treating each setting as one utility--diversity operating point. For a given model, this yields a set of points
$A=\{(U_i,V_i)\}_{i=1}^{n}$,
where $U_i$ and $V_i$ denote the utility and diversity of the $i$-th setting. 
Both metrics are scaled to $[0,1]$, with higher values indicating better performance.
We summarize performance by the non-dominated subset of $A$, denoted by $ND(A)$. A point belongs to $ND(A)$ if no other decoding setting achieves at least as much utility and at least as much diversity, with one of them being strictly better. In other words, $ND(A)$ is the Pareto frontier of the model under the evaluated decoding settings. 

\paragraph{Hypervolume.}
For each experiment, we use a common reference point
$r_{\mathrm{exp}}=(r_U,r_V)$, where $r_U$ and $r_V$ are the minimum utility and diversity observed across all operating points from all compared methods. In two dimensions, the hypervolume of a model is the area of the staircase-shaped region dominated by its non-dominated operating points and bounded below by $r_{\mathrm{exp}}$. Equivalently,
$
HV(A;r_{\mathrm{exp}})
=
\mathrm{Area}\!\left(
\cup_{(U,V)\in ND(A)}
[r_U,U]\times[r_V,V]
\right).
$
Thus, HV is the union area of axis-aligned rectangles induced by the non-dominated set, rather than the area of a single rectangle. Larger HV indicates that the frontier extends further toward high utility and high diversity and/or spans a broader utility--diversity range. Because the reference point is experiment-specific, HV is intended for within-experiment comparison rather than direct comparison across tasks.

\paragraph{Distance to Ideal Point.}
Let \(z^\star=(U^\star, V^\star)\) denote the ideal point, whose coordinates are the best attainable values of the two objectives. For an operating-point set \(A\), we report \(\mathrm{DIP}(A,z^\star)=\min_{(U,V)\in A}\sqrt{(U^\star-U)^2+(V^\star-V)^2}\). Since utility and diversity are scaled to \([0,1]\), we set \(z^\star=(1,1)\). Lower distance is better.

\paragraph{R2 Indicator.}
R2 evaluates an operating-point set under multiple utility-diversity preference weights. For a weight \(\lambda=(\lambda_U,\lambda_V)\), we define the best weighted Tchebycheff shortfall as \(g(A\mid\lambda,z^\star)=\min_{(U,V)\in A}\max\{\lambda_U(z^\star_U-U),\lambda_V(z^\star_V-V)\}\), and compute \(R2(A,\Lambda,z^\star)=\frac{1}{|\Lambda|}\sum_{\lambda\in\Lambda}g(A\mid\lambda,z^\star)\). In our implementation, \(A\) is the full model-specific sweep set and \(\Lambda=\{\lambda^{(\ell)}=(\ell/100,1-\ell/100)\}_{\ell=0}^{100}\). Lower R2 means a smaller average weighted worst-case shortfall to \(z^\star\).

\begin{table*}[!ht]
    \caption{
        \camera{Frontier-level metrics for the three tasks. Each cell reports mean \mbox{\(\pm\)} 95\% confidence interval over five independent sweep runs. HV is higher-is-better, while distance to ideal point (DIP) and R2 are lower-is-better. All small-molecule RL methods use coupled-GRPO, while protein methods use GRPO.}
    }
    \label{tab:frontier_metrics}
    \centering
    \small
    \setlength{\tabcolsep}{3pt}
    \camera{\begin{tabular}{lcccccc}
        \toprule
        \multicolumn{7}{l}{\emph{De novo} Small-Molecule Design}\\
        Metric & Original & GRPO & Mem-GRPO & DMB-GRPO & ER-GRPO & SGRPO\\
        \midrule
        HV \(\uparrow\) & \metricci{0.0686}{0.0030} & \metricci{0.0703}{0.0029} & \metricci{0.0657}{0.0025} & \bestmetricci{0.0774}{0.0030} & \metricci{0.0694}{0.0027} & \metricci{0.0758}{0.0029}\\
        DIP \(\downarrow\) & \metricci{0.2719}{0.0015} & \metricci{0.2677}{0.0013} & \metricci{0.2696}{0.0025} & \metricci{0.2581}{0.0022} & \metricci{0.2677}{0.0019} & \bestmetricci{0.2551}{0.0020}\\
        R2 \(\downarrow\) & \metricci{0.1001}{0.0004} & \metricci{0.1017}{0.0005} & \metricci{0.1055}{0.0002} & \metricci{0.0962}{0.0006} & \metricci{0.0988}{0.0004} & \bestmetricci{0.0955}{0.0004}\\
        \midrule
        \multicolumn{7}{l}{Pocket-Based Small-Molecule Design}\\
        Metric & Original & GRPO & Mem-GRPO & DMB-GRPO & ER-GRPO & SGRPO\\
        \midrule
        HV \(\uparrow\) & \metricci{0.0266}{0.0031} & \metricci{0.0098}{0.0016} & -- & \metricci{0.0219}{0.0023} & \metricci{0.0250}{0.0017} & \bestmetricci{0.0630}{0.0033}\\
        DIP \(\downarrow\) & \metricci{0.4643}{0.0044} & \metricci{0.7459}{0.0153} & -- & \metricci{0.6757}{0.0119} & \metricci{0.6701}{0.0120} & \bestmetricci{0.3797}{0.0022}\\
        R2 \(\downarrow\) & \metricci{0.2381}{0.0022} & \metricci{0.3836}{0.0078} & -- & \metricci{0.3497}{0.0066} & \metricci{0.3457}{0.0062} & \bestmetricci{0.1807}{0.0007}\\
        \midrule
        \multicolumn{7}{l}{\emph{De novo} Protein Design}\\
        Metric & Original & GRPO & Mem-GRPO & DMB-GRPO & ER-GRPO & SGRPO\\
        \midrule
        HV \(\uparrow\) & \metricci{0.2752}{0.0076} & \metricci{0.2140}{0.0054} & \metricci{0.0306}{0.0009} & \metricci{0.0022}{0.0000} & \metricci{0.3128}{0.0052} & \bestmetricci{0.3687}{0.0086}\\
        DIP \(\downarrow\) & \metricci{0.4279}{0.0076} & \metricci{0.6519}{0.0058} & \metricci{1.0128}{0.0011} & \metricci{1.0195}{0.0000} & \metricci{0.5561}{0.0099} & \bestmetricci{0.3538}{0.0048}\\
        R2 \(\downarrow\) & \metricci{0.2201}{0.0027} & \metricci{0.3023}{0.0034} & \metricci{0.4840}{0.0024} & \metricci{0.5145}{0.0001} & \metricci{0.1832}{0.0022} & \bestmetricci{0.1693}{0.0036}\\
        \bottomrule
    \end{tabular}}
\end{table*}

\begin{figure*}[!ht]
    \centering
    \includegraphics[width=\textwidth]{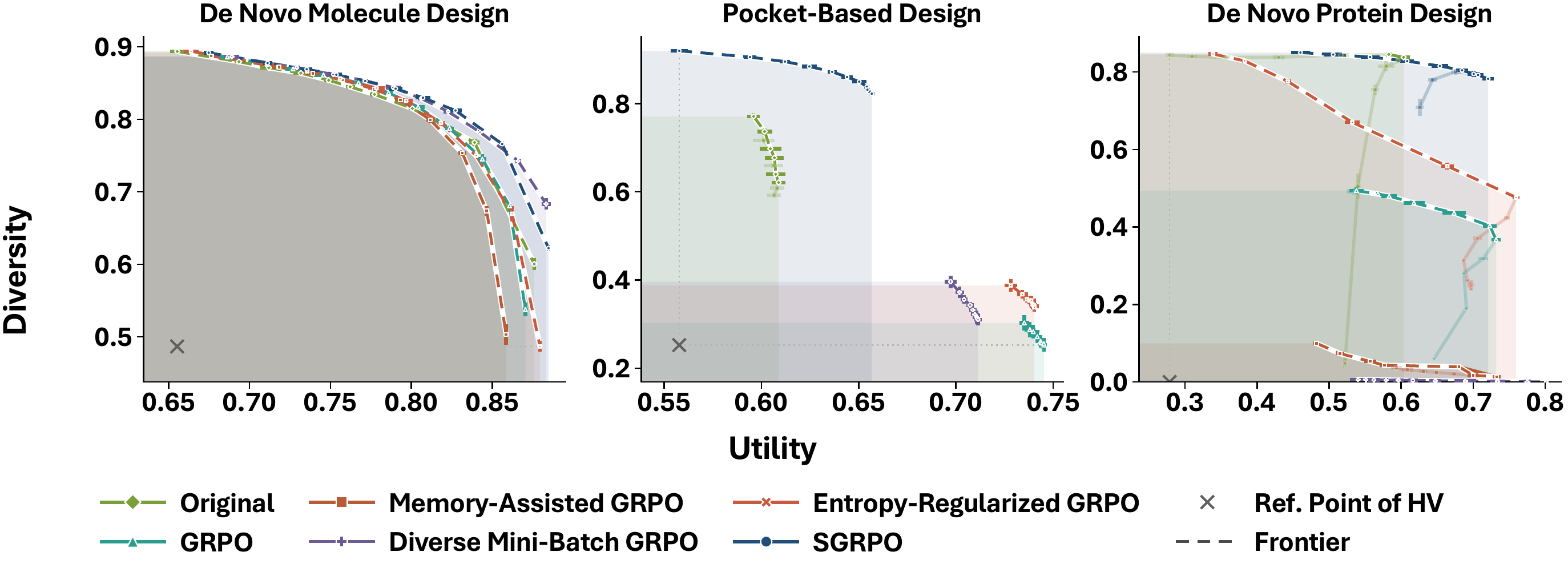}
    \caption{
        Utility--diversity operating points for de novo small-molecule design, pocket-based small-molecule design, and de novo protein design. 
        Each marker corresponds to one decoding setting and reports the mean utility and diversity over five independent \camera{generation sweeps}. Error bars show 95\% confidence intervals on both axes. Dashed lines trace the method-specific non-dominated subsets.
    }
    \label{fig:main_pareto}
\end{figure*}

\subsection{\emph{De novo} Small-Molecule Design}
\label{subsec:denovo_small_molecule}

\paragraph{Setup.}
We study unconditional \emph{de novo} small-molecule design with GenMol~\citep{lee2025genmol}, a discrete diffusion language model that generates molecules as SAFE strings~\citep{noutahi2024gotta}. Unlike autoregressive LMs, GenMol generates samples through iterative denoising, so applying GRPO requires a diffusion-compatible training objective. We therefore instantiate SGRPO on top of coupled-GRPO~\citep{gong2025diffucoder}, which adapts GRPO-style relative policy optimization to discrete diffusion models via coupled denoising samples.
\camera{We use \(K=64\) rollouts per group and \(M=8\) groups per supergroup, matching the 512 rollouts per condition used by GRPO.}
The molecule-level utility in this experiment is defined by drug-likeness and synthetic accessibility. 
We use QED~\citep{bickerton2012quantifying} as a normalized drug-likeness score in $[0,1]$, and denote the raw synthetic accessibility score~\citep{ertl2009estimation} by $\mathrm{SA}(x)$, where lower values indicate easier synthesis. We convert it into a high-is-better score $s_{\mathrm{SA}}(x)=\mathrm{clip}((6-\mathrm{SA}(x))/5,0,1)$, and define the rollout utility as
$
u(x)=\alpha_{\mathrm{QED}}\mathrm{QED}(x)+\alpha_{\mathrm{SA}}s_{\mathrm{SA}}(x).
$
Unless otherwise noted, we use $\alpha_{\mathrm{QED}}=0.6$ and $\alpha_{\mathrm{SA}}=0.4$, slightly prioritizing drug-likeness while retaining synthetic accessibility as a feasibility-oriented component. \camera2{We keep these weights fixed across methods to compare their frontiers under the same utility definition.}
Sample diversity is measured by internal diversity over valid generated molecules using Morgan-fingerprint Tanimoto distances~\citep{bajusz2015tanimoto}. Specifically, \(V(S)=1-\frac{2}{|S|(|S|-1)}\sum_{i<j}s_{ij}\), where \(s_{ij}=\mathrm{Tan}(\phi(x_i),\phi(x_j))\), \(\phi(x)\) is the Morgan fingerprint of molecule \(x\), and \(\mathrm{Tan}\) denotes Tanimoto similarity. 
We compare SGRPO (denoted as coupled-SGRPO) against the pretrained GenMol model, coupled-GRPO, \camera{Memory‑assisted RL-based coupled-GRPO~\citep{blaschke2020memory}, DMB-GRPO, and ER-GRPO.} For Pareto evaluation, each model is decoded under the same sweep of GenMol randomness \(\rho\) and temperature \(\tau\), \camera{using ten paired settings: \((0.1,0.5)\), \((0.2,0.65)\), \((0.3,0.8)\), \((0.4,0.95)\), \((0.5,1.1)\), \((0.6,1.25)\), \((0.7,1.4)\), \((0.8,1.55)\), \((0.9,1.7)\), and \((1.0,2.0)\)}. At each sweep point, we generate 1000 molecules per model and compute both utility metrics and internal diversity over the valid molecules generated at that point.


\paragraph{Result.}
\camera{In Figure~\ref{fig:main_pareto}, SGRPO retains more diversity in the high-utility region than pretrained GenMol, GRPO, Memory-Assisted GRPO, and ER-GRPO. DMB-GRPO yields a similarly strong frontier and retains more diversity near the highest-utility endpoint of the sweep.}
\camera{In Table~\ref{tab:frontier_metrics}, SGRPO achieves the lowest mean DIP (0.2551) and R2 (0.0955). DMB-GRPO has a slightly higher mean HV, \(0.0774\pm0.0030\) versus \(0.0758\pm0.0029\), with overlapping 95\% confidence intervals. SGRPO is therefore competitive with DMB-GRPO on this task.}

\subsection{Pocket-Based Small-Molecule Design}
\paragraph{Setup.}
We train GenMol-P for pocket-based small-molecule design. GenMol-P initializes from the pretrained GenMol and adds pocket-prefix conditioning: a frozen ESM-IF1~\citep{hsu2022learning} pocket encoder embeds pocket, and a two-layer MLP projector maps these embeddings into the GenMol hidden space before molecular denoising. We supervised-tune GenMol-P on the CrossDocked2020~\citep{francoeur2020three} training set.
Following Section~\ref{subsec:denovo_small_molecule}, we instantiate SGRPO as coupled-SGRPO for this discrete diffusion generator. \camera{Coupled-GRPO supplies the diffusion-compatible policy-gradient estimator and serves as the utility-only baseline, playing the role of standard GRPO for an autoregressive generator. We use \(K=24\), \(M=8\), and thus 192 rollouts per pocket for both GRPO and SGRPO.} The rollout utility augments the QED--SA utility with a target-dependent docking term:
$
u(x,\mathcal C)=\alpha_{\mathrm{QED}}\mathrm{QED}(x)+\alpha_{\mathrm{SA}}s_{\mathrm{SA}}(x)+\alpha_{\mathrm{dock}}s_{\mathrm{dock}}(x,\mathcal C),
$
where
$
s_{\mathrm{dock}}(x,\mathcal C)=\mathrm{clip}(-a_{\mathrm{Vina}}(x,\mathcal C)/10,0,1)
$
maps the raw AutoDock Vina~\citep{trott2010autodock} score $a_{\mathrm{Vina}}(x,\mathcal C)$ to a high-is-better score in $[0,1]$. Unless otherwise noted, we set $\alpha_{\mathrm{QED}}=0.3$, $\alpha_{\mathrm{SA}}=0.2$, and $\alpha_{\mathrm{dock}}=0.5$, giving docking the largest weight because pocket compatibility is the primary task objective, while QED and SA act as chemistry-oriented regularizers. 
We do not retune these coefficients per method, \camera2{ensuring a shared utility definition across comparisons.}
We compare coupled-SGRPO against \camera{the original GenMol-P model, coupled-GRPO, DMB-GRPO, and ER-GRPO} under the same \camera{ten} paired $(\rho,\tau)$ settings used in Section~\ref{subsec:denovo_small_molecule}, which span conservative to exploratory decoding regimes, using pockets from the CrossDocked2020 test set.
\footnote{We exclude Memory-assisted coupled-GRPO here: global memory mixes unrelated pockets, while pocket-specific memories require separate optimization and extra compute.}
At each sweep point, each model generates 16 ligands for each of 100 held-out pockets, for 1600 ligand samples in total. This protocol provides stable estimates at a manageable docking cost. Utility metrics are averaged over the valid generated ligands at that point, while diversity is computed within each pocket's 16 ligands and then averaged over pockets. This measures target-conditional diversity rather than conflating diversity with variation across different pockets.
\camera{Specifically, for the valid ligand collection \(S_{\mathcal C}\) of pocket \(\mathcal C\), \(V(S_{\mathcal C})=1-\frac{2}{|S_{\mathcal C}|(|S_{\mathcal C}|-1)}\sum_{i<j}\mathrm{Tan}(\phi(x_i),\phi(x_j))\), where \(\phi\) is the Morgan fingerprint.}

\paragraph{Result.}
\camera{Pocket-based design is the setting where SGRPO helps most: because optimizing docking to a fixed pocket tends to collapse generation onto a few high-scoring chemotypes, coupled-GRPO, DMB-GRPO, and ER-GRPO improve utility only by sacrificing within-pocket diversity, whereas coupled-SGRPO shifts the GenMol-P operating-point trajectory outward relative to baselines and retains markedly higher diversity at comparable utility, especially in the high-utility regime where collapse is strongest (Figure~\ref{fig:main_pareto}). This matches the conditional nature of the task: multiple chemically distinct ligands can be similarly plausible for the same pocket, but utility-only relative updates over-amplify small score differences and reinforce redundancy, while SGRPO explicitly rewards high-utility samples that also contribute marginal diversity.} 
Accordingly, coupled-SGRPO achieves the best frontier-level performance in Table~\ref{tab:frontier_metrics}, with the largest HV and the smallest DIP and R2, and its advantage is more pronounced than in \emph{de novo} molecule generation. Figure~\ref{fig:main_pareto} shows all evaluated decoding settings for completeness, so some adjacent baseline points may improve in both utility and diversity before the true trade-off boundary is reached.

\subsection{\emph{De novo} Protein Design}
\label{subsec:denovo_protein}

\paragraph{Setup.}
We evaluate unconditional \emph{de novo} protein design with ProGen2~\citep{nijkamp2023progen2}, an autoregressive amino-acid language model, and apply SGRPO via GRPO~\citep{shao2024deepseekmath}. The sequence-level utility combines naturalness, foldability, stability, and developability:
\(u(y)=\alpha_{\mathrm{nat}}r_{\mathrm{nat}}(y)+\alpha_{\mathrm{fold}}r_{\mathrm{fold}}(y)+\alpha_{\mathrm{stab}}r_{\mathrm{stab}}(y)+\alpha_{\mathrm{dev}}r_{\mathrm{dev}}(y)\).
Here, the four terms are computed from ESM2, ESMFold~\citep{lin2023evolutionary}, TemBERTure~\citep{rodella2024temberture}, and ProteinSol-based scorers~\citep{hebditch2017protein}, respectively, with weights \(\alpha_{\mathrm{nat}}=0.25\), \(\alpha_{\mathrm{fold}}=0.30\), \(\alpha_{\mathrm{stab}}=0.20\), and \(\alpha_{\mathrm{dev}}=0.25\). Diversity is measured at the set level using normalized Levenshtein similarity over valid sequences:
\(V(S)=1-\frac{2}{|S|(|S|-1)}\sum_{i<j}s^{\mathrm{edit}}_{ij}\).
\camera{SGRPO uses \(K=12\) and \(M=8\), matching the 96 rollouts per condition in GRPO.}
We compare SGRPO against \camera{the original ProGen2 model, GRPO, Memory-assisted RL-based GRPO, DMB-GRPO, and ER-GRPO} under the same temperature sweep \(\tau\in\{0.1,0.2,\ldots,1.0,1.1,1.2\}\). At each temperature, we sample 512 sequences per model and evaluate both utility and diversity over valid outputs.

\paragraph{Result.}
SGRPO achieves the best utility--diversity trade-off in \emph{de novo} protein design. \camera{In Figure~\ref{fig:main_pareto}, GRPO loses substantial diversity, Memory-Assisted GRPO loses still more, and DMB-GRPO exhibits near-zero diversity across the sweep. ER-GRPO preserves more diversity than GRPO but remains below SGRPO in the high-utility region. SGRPO reaches high utility while retaining diversity much closer to the pretrained model.} This pattern is reflected consistently in Table~\ref{tab:frontier_metrics}, where SGRPO attains the best HV, DIP, and R2.

%% file: sec/analysis.tex
\section{Analysis}

\subsection{Ablation Study}
\label{subsec:ablations}

We isolate two components of SGRPO on the GenMol-based \emph{de novo} small-molecule design task: the supergroup diversity reward and the leave-one-out diversity contribution mechanism.

\begin{wrapfigure}[19]{r}{0.4\columnwidth}
    \centering
    \includegraphics[width=\linewidth]{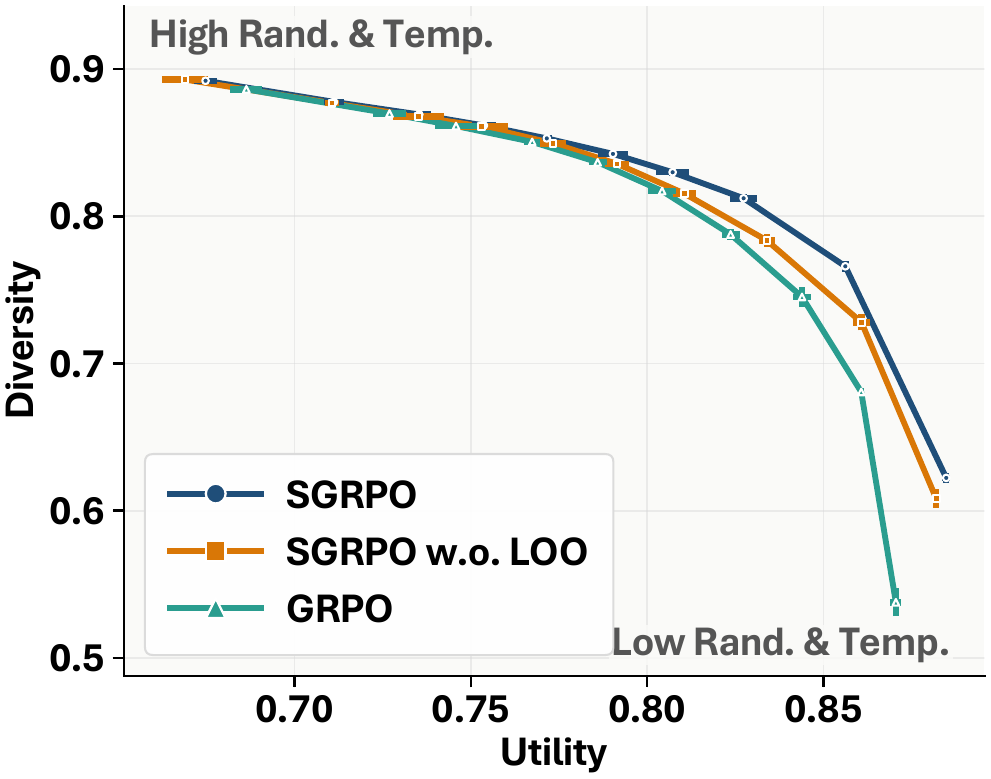}
  
    \caption{
    \small{
        Ablation study on GenMol-based \emph{de novo} small-molecule generation. Each point reports the mean over five independent sweeps, and error bars indicate 95\% confidence intervals for utility and diversity. Removing the diversity term yields coupled-GRPO, while removing leave-one-out group credit weakens set-aware credit assignment.}
    }
    \label{fig:ablation}

\end{wrapfigure}

As shown in Figure~\ref{fig:ablation}, removing the diversity reward gives the innermost utility-diversity curve. Adding the supergroup diversity reward without leave-one-out credit assignment moves the curve outward, indicating that group-level diversity pressure itself improves the trade-off. Full SGRPO further dominates this variant, showing that leave-one-out credit assignment is important for redistributing the group diversity reward to the rollouts that actually contribute to set-level diversity.

\subsection{Training Dynamics of Generated Distributions}

To understand why the three training algorithms lead to different final utility-diversity trade-offs, we use \emph{de novo} protein design as a diagnostic setting and visualize how their generated sequence distributions move during training. For GRPO, Memory-assisted GRPO, and SGRPO, we sample 512 sequences from the shared original ProGen2 model, the checkpoint after 20 optimization steps, and the final checkpoint after 100 steps. We then pool all sequences and compute a shared two-dimensional UMAP embedding using distances derived from normalized Levenshtein similarity, so that movements are comparable across methods and checkpoints.

\begin{wrapfigure}[28]{r}{0.45\columnwidth}
    \centering

    \includegraphics[width=\linewidth]{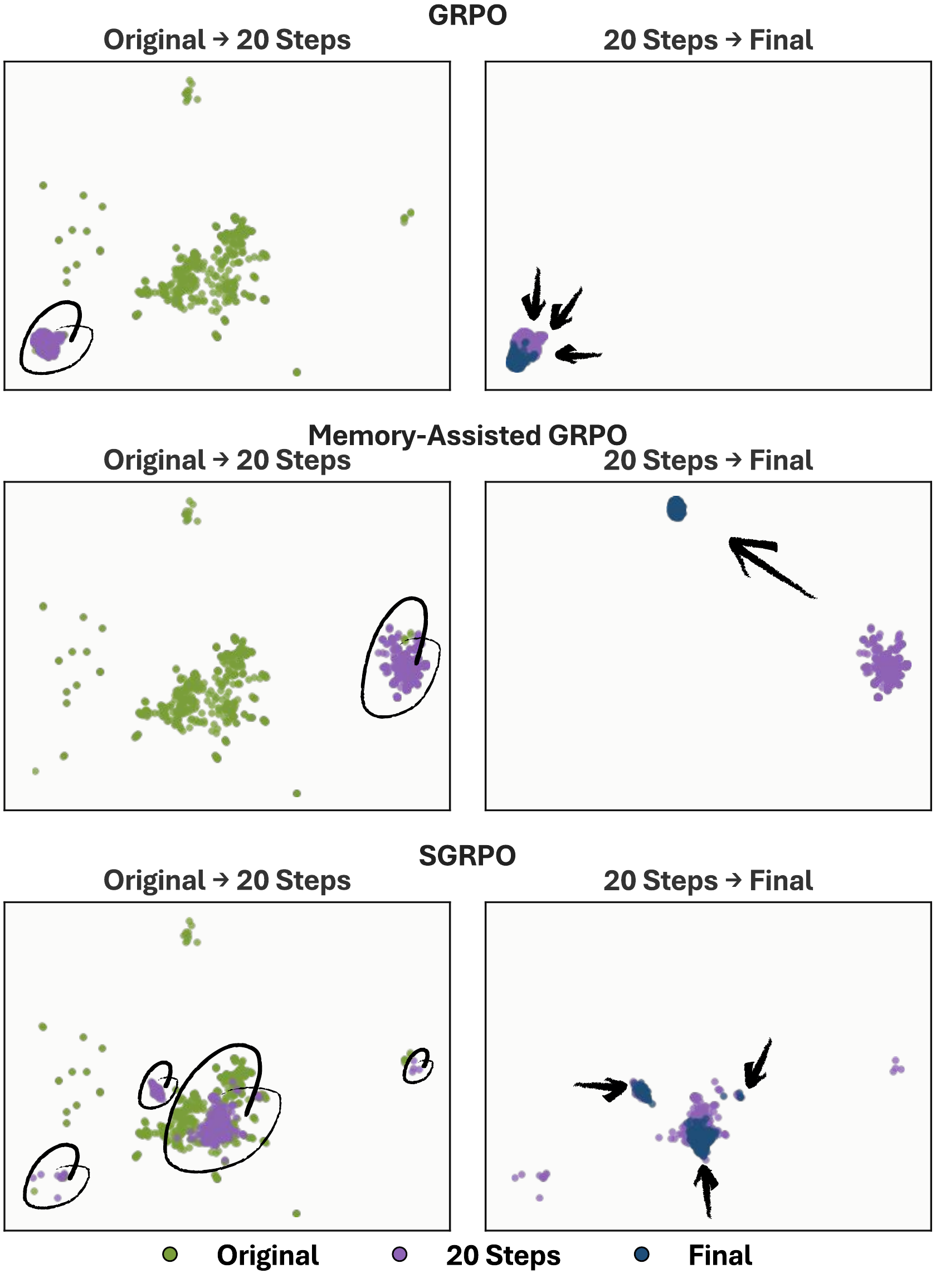}
    \caption{
        Distribution dynamics during ProGen2 post-training. SGRPO explores multiple clusters early and preserves them, whereas GRPO contracts and Memory-assisted GRPO drifts toward a narrow distant region.
    }
    \label{fig:progen2_distribution_dynamics}
\end{wrapfigure}

Figure~\ref{fig:progen2_distribution_dynamics} suggests two distinct training dynamics. After 20 steps, GRPO and Memory-assisted GRPO each move into a relatively concentrated region of the embedding, whereas SGRPO spreads across multiple clusters, including the regions explored by the other two methods. This early behavior indicates that supergroup-relative diversity pressure promotes broader exploration. From 20 to 100 steps, GRPO contracts into an even smaller region. Memory-assisted GRPO instead drifts toward a distant, narrow region: its memory penalty discourages revisiting high-density regions, which is not equivalent to directly optimizing set diversity and can drive distributional drift. SGRPO refines within the explored regions while retaining multiple clusters, explaining why it reaches the high-utility regime while preserving substantially more sequence diversity.

\subsection{Robustness to Diversity-Estimator Efficiency and Reward Weighting}
\begin{figure}[!ht]
    \centering
    \small
    \includegraphics[width=\textwidth]{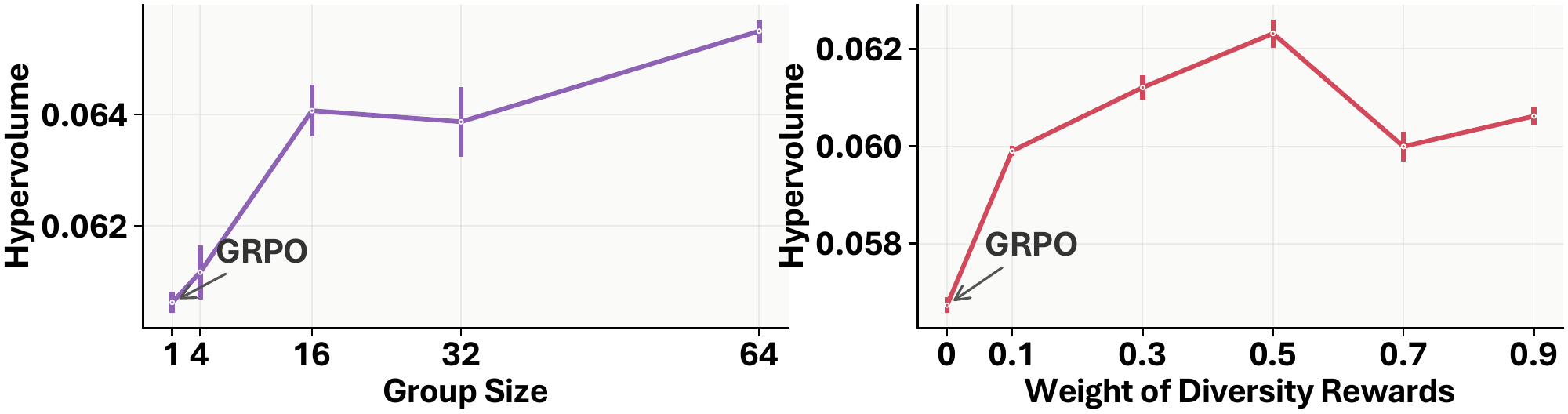}
    \caption{
        Analysis of diversity-estimator efficiency and group-reward weighting on the GenMol-based \emph{de novo} small-molecule design task. Each point reports the mean HV over five independent sweeps, and error bars indicate 95\% confidence intervals. Left: under a fixed total supergroup size, SGRPO already outperforms GRPO with small groups, indicating that useful group-diversity estimates do not require a large rollout multiplier. Right: SGRPO robustly improves HV over GRPO across all tested nonzero \(\lambda\), with the strongest frontier expansion at \(\lambda=0.5\).
    }
    \label{fig:hyperparam}
\end{figure}

We study two practical sensitivities of SGRPO on the GenMol-based \emph{de novo} small-molecule design task: the efficiency of the group-diversity estimator and the choice of group-reward weight \(\lambda\). In both cases, we train one model per setting, evaluate each model using the same decoding sweep as in Section~\ref{subsec:denovo_small_molecule}, and measure the HV of the resulting utility-diversity frontier using a common lower reference point computed from all operating points in this experiment. To test tolerance to diversity-estimator inefficiency, we hold the total supergroup size fixed while varying the group size \(K\in\{1,4,16,32,64\}\), where \(K=1\) recovers GRPO. As shown in the left panel of Figure~\ref{fig:hyperparam}, SGRPO already outperforms GRPO in the near-minimal setting \mbox{\(K=4\)}, indicating that useful group-diversity estimates do not require a large multiplicative increase in rollouts. Larger groups generally lead to better frontier-level performance, suggesting that more efficient diversity estimates help but are not necessary for SGRPO to be effective.

We next vary the group-reward coefficient \(\lambda\in\{0,0.1,0.3,0.5,0.7,0.9\}\), which controls the trade-off between rollout-level utility and redistributed group-level diversity in the composed reward. The right panel of Figure~\ref{fig:hyperparam} shows that every tested nonzero value of \(\lambda\) improves HV over the GRPO baseline \(\lambda=0\), demonstrating that the benefit of supergroup-relative diversity pressure is robust to reward weighting. Performance peaks at \(\lambda=0.5\), suggesting that the strongest frontier expansion is achieved by balancing molecule-level utility with set-level diversity. Overall, Figure~\ref{fig:hyperparam} shows that SGRPO is robust to both imperfect diversity estimation and moderate variation in reward composition.

%% file: sec/conclusion.tex
\section{Conclusion}

We presented SGRPO, a GRPO-style framework for directly combining rollout-level utility with set-level sample diversity in biomolecular post-training. The central idea is simple: sample multiple candidate sets under the same condition, score their diversity, compare them within the supergroup, and redistribute the resulting set reward to individual rollouts through leave-one-out diversity contributions. Across \emph{de novo} small-molecule design, pocket-based small-molecule design, and \emph{de novo} protein design, SGRPO improves the utility-diversity Pareto frontier over pretrained generators and RL baselines.
\camera{This benefit comes with additional computational overhead, which we mitigate by reusing pairwise similarities for diversity rewards and leave-one-out credit assignment (Appendix~\ref{app:cost_details}).}
These results suggest that directly rewarding diverse generated sets is a practical way to expand the operating points available to biomolecular generation models.

%% file: sec/appendix.tex
\section{Full Training Procedure of SGRPO}
\label{app:full_algorithm}

This appendix provides the full training procedure of Supergroup Relative Policy Optimization (SGRPO), corresponding to Section~\ref{sec:sgrpo} in the main text. For each condition, SGRPO samples a same-condition supergroup, computes rollout-level utility and group-level diversity, redistributes the diversity signal to individual rollouts in a set-aware manner, forms supergroup-relative advantages, and updates the policy with a PPO-style objective regularized toward a reference policy.

\begin{figure*}[!ht]
\centering
\footnotesize
\begin{minipage}{0.985\textwidth}
\hrule
\algcaption{Supergroup Relative Policy Optimization (SGRPO)}\label{alg:sgrpo}
\hrule

\begin{algorithmic}[1]
\STATE \textbf{Input:} initial policy $\pi_\theta$, reference policy $\pi_{\mathrm{ref}}$, condition batch $\{\mathcal C_b\}_{b=1}^B$, groups per condition $M$, rollouts per group $K$, diversity weight $\lambda$, PPO clip parameter $\epsilon$, KL coefficient $\beta$, contribution temperature $\tau_c$
\FOR{$t=1$ to $T$}
    \STATE Set $\theta_{\mathrm{old}} \leftarrow \theta$
    \FOR{$b=1$ to $B$}
        \STATE Sample a same-condition supergroup $\mathcal S_b=\{G_{b,1},\dots,G_{b,M}\}$, where each group is $G_{b,m}=\{x_{b,m,1},\dots,x_{b,m,K}\}$ with $x_{b,m,i}\sim \pi_{\theta_{\mathrm{old}}}(\cdot\mid\mathcal C_b)$
        \STATE Compute rollout utilities $r_{b,m,i}=r(x_{b,m,i},\mathcal C_b)$ for all $m,i$
        \STATE Compute group diversities $R_{b,m}=D(G_{b,m})$ for all $m$
        \STATE Compute the supergroup mean diversity $\bar R_b=\frac{1}{M}\sum_{h=1}^{M}R_{b,h}$
        \STATE Compute group-relative diversity signals
        $
        A^{\mathrm{grp}}_{b,m}
        =
        \frac{M}{M-1}(R_{b,m}-\bar R_b)
        $
        for all $m$
        \FOR{$m=1$ to $M$}
            \STATE For each rollout $x_{b,m,i}$, compute leave-one-out contribution
            $
            c_{b,m,i}
            =
            D(G_{b,m})-D(G_{b,m}\setminus\{x_{b,m,i}\})
            $
            \STATE Compute within-group standardized contributions
            $
            z_{b,m,i}
            =
            \dfrac{c_{b,m,i}-\bar c_{b,m}}{\sigma(c_{b,m,\cdot})+\zeta}
            $
            for all $i$
            \STATE Form sign-aware redistribution weights
            $
            w^{\pm}_{b,m,i}
            =
            K\,
            \dfrac{\exp(\pm z_{b,m,i}/\tau_c)}
            {\sum_{j=1}^{K}\exp(\pm z_{b,m,j}/\tau_c)}
            $
            for all $i$
            \STATE Assign redistributed diversity rewards
            $
            \widetilde R_{b,m,i}
            =
            R_{b,m}
            +
            [A^{\mathrm{grp}}_{b,m}]_+(w^+_{b,m,i}-1)
            -
            [-A^{\mathrm{grp}}_{b,m}]_+(w^-_{b,m,i}-1)
            $
            for all $i$
        \ENDFOR
        \STATE Form combined rollout rewards
        $
        \widehat r_{b,m,i}
        =
        (1-\lambda)r_{b,m,i}
        +
        \lambda \widetilde R_{b,m,i}
        $
        for all $m,i$
        \STATE Compute the supergroup mean reward
        $
        \bar r_{\mathcal S_b}
        =
        \frac{1}{MK}\sum_{m=1}^{M}\sum_{i=1}^{K}\widehat r_{b,m,i}
        $
        \STATE Compute supergroup-relative advantages
        $
        A_{b,m,i}
        =
        \frac{MK}{MK-1}
        \bigl(\widehat r_{b,m,i}-\bar r_{\mathcal S_b}\bigr)
        $
        for all $m,i$
    \ENDFOR
    \STATE Update $\theta$ by minimizing the PPO-style objective with KL regularization to $\pi_{\mathrm{ref}}$ using all collected tuples $(\mathcal C_b,x_{b,m,i},A_{b,m,i})$
\ENDFOR
\STATE \textbf{Return:} trained policy $\pi_\theta$
\end{algorithmic}

\hrule
\end{minipage}
\end{figure*}

\paragraph{Implementation notes.}
All relative comparisons in SGRPO are performed within the same-condition supergroup. In particular, group-relative diversity signals are centered only across the $M$ groups sampled for the same condition, and supergroup-relative advantages are centered only across the corresponding $MK$ rollouts. This avoids confounding policy quality with variation in condition difficulty. In practice, the diversity metric $D(\cdot)$ can be instantiated according to the domain, and the utility reward $r(x,\mathcal C)$ can be any task-specific scalar oracle.

\section{\camera{Population Diversity and Finite-Group Estimation}}
\label{app:sgrpo_theory}

\camera{We first establish that finite-group diversity is an unbiased estimator of policy-level expected diversity for normalized fixed-order U-statistics, with partition consistency as a corresponding finite-sample property. We then study concentration for bounded pairwise diversity. These are properties of the diversity estimator. The effect of candidate-wise reward redistribution on policy optimization is treated separately in Appendix~\ref{app:redistribution_objective}.}

\paragraph{Setup.}
\camera{Samples below are indexed collections, retaining repeated generated candidates. Partitions act on sample indices. We assume \(K\geq 2\) for pairwise diversity, \(M\geq 1\) for the estimation results, and a symmetric similarity function. The SGRPO group-relative baseline additionally requires \(M\geq 2\).}
Let $\mathcal{X}=\{x_1,\dots,x_N\}$ be a same-condition sample of size $N=MK$, where $M$ is the number of groups and $K$ is the group size. A \emph{balanced partition} $\Pi$ divides $\mathcal{X}$ into $M$ disjoint groups of size $K$:
\[
\Pi=\{G_1,\dots,G_M\},
\qquad
|G_m|=K,
\qquad
\bigsqcup_{m=1}^M G_m=\mathcal{X}.
\]
For a set $A=\{a_1,\dots,a_L\}$, define the normalized pairwise diversity
\begin{equation}
D_L(A)
=
\frac{2}{L(L-1)}
\sum_{1\le i<j\le L}
d(a_i,a_j),
\qquad
d(x,x') := 1-s(x,x'),
\label{eq:app_normalized_pairwise_diversity}
\end{equation}
where $s(x,x')\in[0,1]$ is a biomolecular similarity function. Equivalently, $D_L(A)$ is the average pairwise dissimilarity under $d\in[0,1]$.

Given a balanced partition $\Pi=\{G_1,\dots,G_M\}$, define the average small-group diversity
\begin{equation}
\overline D_{M,K}(\mathcal{X},\Pi)
=
\frac{1}{M}
\sum_{m=1}^M D_K(G_m).
\label{eq:app_avg_small_group_diversity}
\end{equation}

\subsection{\camera{Population Unbiasedness and Partition Consistency}}

\camera{
\begin{proposition}[Population unbiasedness for fixed-order diversity]
\label{prop:population_partition_consistency}
Let \(h\) be a symmetric, integrable kernel of fixed order \(q\leq K\). For an indexed sample of size \(L\geq q\), define
\begin{equation}
D_L^{(q)}(X_{1:L})
=\binom{L}{q}^{-1}
\sum_{\substack{I\subseteq\{1,\ldots,L\}\\ |I|=q}}h(X_I),
\qquad
\mu_q(\theta,\mathcal C)
=\mathbb E_{X_1,\ldots,X_q\overset{\mathrm{i.i.d.}}{\sim}\pi_\theta(\cdot\mid\mathcal C)}
[h(X_1,\ldots,X_q)].
\label{eq:app_u_stat_diversity}
\end{equation}
For \(N=MK\) i.i.d.\ samples under a fixed condition \(\mathcal C\), partitioned independently and uniformly into \(M\) groups of size \(K\), every group satisfies
\begin{equation}
\mathbb E[D_K^{(q)}(G_m)\mid\mathcal C]
=\mathbb E[D_N^{(q)}(X_{1:N})\mid\mathcal C]
=\mu_q(\theta,\mathcal C).
\label{eq:app_population_unbiasedness}
\end{equation}
For any fixed realization \(x_{1:N}\), the corresponding partition identities are
\begin{equation}
\mathbb E_\Pi[D_K^{(q)}(G_m)\mid x_{1:N}]
=\mathbb E_\Pi\!\left[\frac{1}{M}\sum_{h=1}^M D_K^{(q)}(G_h)\,\middle|\,x_{1:N}\right]
=D_N^{(q)}(x_{1:N}).
\label{eq:app_general_partition}
\end{equation}
\end{proposition}

\begin{proof}
Each \(q\)-tuple of distinct sample indices has the same expected kernel value \(\mu_q(\theta,\mathcal C)\), which proves population unbiasedness by linearity of expectation. For a fixed realization, a given \(q\)-subset lies in a particular group with probability \(\binom{K}{q}/\binom{N}{q}\), and in any one of the \(M\) groups with probability \(M\binom{K}{q}/\binom{N}{q}\). These probabilities cancel the within-group normalization and, for the aggregate identity, the factor \(1/M\). Equivalently, exchangeability of the groups gives the per-group identity from the aggregate one.
\end{proof}

For pairwise diversity, \(q=2\) and \(h(x,x')=1-s(x,x')\). Thus the group reward is unbiased for the expected diversity of independent policy samples, not merely for a particular sampled supergroup. Sampling a supergroup and partitioning it independently is distributionally equivalent to sampling the groups directly. The original pairwise partition calculation is given below as a special case.
}

\begin{proposition}[Partition consistency of normalized pairwise diversity]
\label{prop:partition_consistency}
Let $\mathcal{X}=\{x_1,\dots,x_N\}$ be any fixed set of size $N=MK$. If $\Pi$ is drawn uniformly from all balanced partitions of $\mathcal{X}$ into $M$ groups of size $K$, then
\begin{equation}
\mathbb{E}_{\Pi}\!\left[\overline D_{M,K}(\mathcal{X},\Pi)\right]
=
D_N(\mathcal{X}).
\label{eq:app_partition_consistency}
\end{equation}
\end{proposition}

\begin{proof}
For brevity, let $d_{ij}:=d(x_i,x_j)$. For a random balanced partition $\Pi$, define the indicator
\[
I_{ij}(\Pi)
=
\mathbf{1}\{x_i \text{ and } x_j \text{ belong to the same group under } \Pi\}.
\]
Then
\begin{align}
\overline D_{M,K}(\mathcal{X},\Pi)
&=
\frac{1}{M}
\sum_{m=1}^M
\frac{2}{K(K-1)}
\sum_{\substack{i<j\\ x_i,x_j\in G_m}}
d_{ij}
\nonumber\\
&=
\frac{1}{M}\cdot \frac{2}{K(K-1)}
\sum_{1\le i<j\le N}
I_{ij}(\Pi)\, d_{ij}.
\label{eq:app_partition_indicator_form}
\end{align}
Taking expectation over $\Pi$ gives
\begin{equation}
\mathbb{E}_{\Pi}\!\left[\overline D_{M,K}(\mathcal{X},\Pi)\right]
=
\frac{1}{M}\cdot \frac{2}{K(K-1)}
\sum_{1\le i<j\le N}
\mathbb{E}_{\Pi}[I_{ij}(\Pi)]\, d_{ij}.
\label{eq:app_partition_expectation_expand}
\end{equation}

By symmetry of the uniform balanced partition, for any fixed pair $(i,j)$,
\begin{equation}
\mathbb{P}_{\Pi}\!\left(I_{ij}(\Pi)=1\right)
=
\frac{K-1}{N-1},
\label{eq:app_same_group_probability}
\end{equation}
because once $x_i$ is assigned to a group, exactly $K-1$ of the remaining $N-1$ positions lie in that same group. Substituting Eq.~\eqref{eq:app_same_group_probability} into Eq.~\eqref{eq:app_partition_expectation_expand},
\begin{align}
\mathbb{E}_{\Pi}\!\left[\overline D_{M,K}(\mathcal{X},\Pi)\right]
&=
\frac{1}{M}\cdot \frac{2}{K(K-1)}\cdot \frac{K-1}{N-1}
\sum_{1\le i<j\le N} d_{ij}
\nonumber\\
&=
\frac{2}{MK(N-1)}
\sum_{1\le i<j\le N} d_{ij}
\nonumber\\
&=
\frac{2}{N(N-1)}
\sum_{1\le i<j\le N} d_{ij}
=
D_N(\mathcal{X}),
\end{align}
where we used $N=MK$ in the last step.
\end{proof}

\camera{Combining the finite-sample identity with policy sampling gives \(\mathbb E_{\mathcal X,\Pi}[D_K(G_m)\mid\mathcal C]=\mathbb E_{\mathcal X}[D_N(\mathcal X)\mid\mathcal C]=\mu_2(\theta,\mathcal C)\). This identifies the population quantity estimated by each group reward.}

\subsection{Concentration Around Full-Sample Diversity}

We now show that the average small-group diversity not only matches the full-sample diversity in expectation, but also concentrates around it when the groups are sufficiently large.

\begin{proposition}[Concentration of average small-group diversity]
\label{prop:small_group_concentration}
Fix a condition $c$, and let $X_1,\dots,X_N \overset{\mathrm{i.i.d.}}{\sim} \pi_\theta(\cdot \mid c)$ with $N=MK$. Let $\mathcal{X}=\{X_1,\dots,X_N\}$, and let $\Pi$ be a uniformly random balanced partition of $\mathcal{X}$ into $M$ groups of size $K$, independent of the samples. Define
\begin{equation}
\mu(c)
:=
\mathbb{E}_{X,X'\sim\pi_\theta(\cdot\mid c)}
\bigl[d(X,X')\bigr].
\label{eq:app_population_diversity}
\end{equation}
Then, for any $\varepsilon>0$,
\begin{equation}
\mathbb{P}\!\left(
\left|
\overline D_{M,K}(\mathcal{X},\Pi)-D_N(\mathcal{X})
\right|
\ge \varepsilon
\right)
\le
4\exp\!\left(
-\frac{1}{2}
M\left\lfloor \frac{K}{2}\right\rfloor \varepsilon^2
\right).
\label{eq:app_concentration_tail}
\end{equation}
Consequently, for any $\delta\in(0,1)$,
\begin{equation}
\left|
\overline D_{M,K}(\mathcal{X},\Pi)-D_N(\mathcal{X})
\right|
\le \varepsilon
\end{equation}
with probability at least $1-\delta$ whenever
\begin{equation}
M\left\lfloor \frac{K}{2}\right\rfloor
\ge
\frac{2\log(4/\delta)}{\varepsilon^2}.
\label{eq:app_group_size_condition}
\end{equation}
Equivalently, a sufficient lower bound on the group size is
\begin{equation}
K
\ge
2\left\lceil
\frac{2\log(4/\delta)}{M\varepsilon^2}
\right\rceil.
\label{eq:app_k_lower_bound}
\end{equation}
\end{proposition}

\begin{proof}
Both $D_N(\mathcal{X})$ and $\overline D_{M,K}(\mathcal{X},\Pi)$ estimate the same population quantity $\mu(c)$ defined in Eq.~\eqref{eq:app_population_diversity}.

First, $D_N(\mathcal{X})$ is a bounded U-statistic of order two with kernel $d(\cdot,\cdot)\in[0,1]$. By Hoeffding's concentration inequality for bounded U-statistics,
\begin{equation}
\mathbb{P}\!\left(
\left|D_N(\mathcal{X})-\mu(c)\right|\ge t
\right)
\le
2\exp\!\left(
-2\left\lfloor \frac{N}{2}\right\rfloor t^2
\right).
\label{eq:app_full_u_stat_concentration}
\end{equation}

Next, consider $\overline D_{M,K}(\mathcal{X},\Pi)$. Because the samples are i.i.d.\ and the partition is independent of the samples, its distribution is the same as that obtained by first drawing $MK$ i.i.d.\ samples and then forming $M$ disjoint blocks of size $K$. Hence
\[
\overline D_{M,K}(\mathcal{X},\Pi)
=
\frac{1}{M}\sum_{m=1}^M U_m,
\]
where $U_1,\dots,U_M$ are independent copies of a bounded order-two U-statistic based on $K$ i.i.d.\ samples with mean $\mu(c)$. For each $U_m$, Hoeffding's inequality yields
\[
\mathbb{P}\!\left(
|U_m-\mu(c)|\ge t
\right)
\le
2\exp\!\left(
-2\left\lfloor \frac{K}{2}\right\rfloor t^2
\right).
\]
\camera{To retain the stated constants when averaging groups, let \(n_K=\lfloor K/2\rfloor\). A pairwise U-statistic is the average, over permutations of its sample indices, of averages of \(n_K\) disjoint pair terms. Each such average consists of independent variables in \([0,1]\). Jensen's inequality and Hoeffding's lemma therefore give \(\mathbb E\exp\{a(U_m-\mu(c))\}\leq\exp\{a^2/(8n_K)\}\). Independence across groups implies \(\mathbb E\exp\{a(\frac1M\sum_m U_m-\mu(c))\}\leq\exp\{a^2/(8Mn_K)\}\). Applying the Chernoff bound to both tails yields}
\begin{equation}
\mathbb{P}\!\left(
\left|
\overline D_{M,K}(\mathcal{X},\Pi)-\mu(c)
\right|
\ge t
\right)
\le
2\exp\!\left(
-2M\left\lfloor \frac{K}{2}\right\rfloor t^2
\right).
\label{eq:app_small_group_u_stat_concentration}
\end{equation}

Applying the triangle inequality,
\begin{equation}
\left|
\overline D_{M,K}(\mathcal{X},\Pi)-D_N(\mathcal{X})
\right|
\le
\left|
\overline D_{M,K}(\mathcal{X},\Pi)-\mu(c)
\right|
+
\left|
D_N(\mathcal{X})-\mu(c)
\right|.
\label{eq:app_triangle_ineq}
\end{equation}
Setting $t=\varepsilon/2$ in Eqs.~\eqref{eq:app_full_u_stat_concentration} and \eqref{eq:app_small_group_u_stat_concentration}, and then using a union bound, we obtain
\begin{align}
&\mathbb{P}\!\left(
\left|
\overline D_{M,K}(\mathcal{X},\Pi)-D_N(\mathcal{X})
\right|
\ge \varepsilon
\right)
\nonumber\\
&\le
2\exp\!\left(
-\frac{1}{2}
M\left\lfloor \frac{K}{2}\right\rfloor \varepsilon^2
\right)
+
2\exp\!\left(
-\frac{1}{2}
\left\lfloor \frac{N}{2}\right\rfloor \varepsilon^2
\right).
\label{eq:app_union_bound_two_terms}
\end{align}
Since $N=MK$ and $\lfloor N/2\rfloor \ge M\lfloor K/2\rfloor$, the second term is no larger than the first, which gives Eq.~\eqref{eq:app_concentration_tail}. Finally, solving
\[
4\exp\!\left(
-\frac{1}{2}
M\left\lfloor \frac{K}{2}\right\rfloor \varepsilon^2
\right)
\le \delta
\]
yields the sufficient condition in Eq.~\eqref{eq:app_group_size_condition}, and Eq.~\eqref{eq:app_k_lower_bound} follows immediately.
\end{proof}

\camera{The concentration statement concerns the average group estimator and probability over both policy samples and the independent partition. At fixed \(M\), its bound improves with \(K\). At fixed total rollout count \(N=MK\), \(M\lfloor K/2\rfloor\) is approximately \(N/2\), so this bound alone does not explain the group-size effect in the fixed-budget ablation. A single group's diversity has the \(K\)-dependent bound above.}

\camera{\paragraph{Scope beyond pairwise diversity.}
The fixed-order result applies to normalized U-statistics, but empirical entropy, coverage, and relative scaffold diversity are generally outside this class. Their expectations can depend on group size, so exact alignment with full-sample diversity is not implied. SGRPO can nevertheless use a computable set score whenever both the group and its leave-one-out subsets can be scored. This algorithmic compatibility is distinct from the metric-specific unbiasedness guarantee. Appendix~\ref{app:scaffold_training} evaluates a non-pairwise example.}

\input{sec/revision_objective}

\section{GenMol-P Implementation}
\label{sec:app_genmol_p}

\subsection{Method}

GenMol-P extends GenMol~\citep{lee2025genmol} from unconditional molecular generation to pocket-conditioned molecular generation by adding a continuous pocket prefix to the discrete diffusion language model. Let \(x=(x_1,\ldots,x_T)\) denote the SAFE-token sequence of a ligand and let \(\mathcal C\) denote a protein pocket. In GenMol-P, \(\mathcal C\) is represented by the residue sequence together with backbone coordinates \((N,C_\alpha,C)\) for each pocket residue. A frozen ESM-IF1 encoder~\citep{hsu2022learning} maps these pocket backbone coordinates to residue-level embeddings \(h_1,\ldots,h_L\), where \(L\) is the number of pocket residues. A trainable two-layer projector \(P_\psi\) then maps each residue embedding into the GenMol hidden space, producing prefix vectors \(p_\ell=P_\psi(h_\ell)\in\mathbb R^d\), with \(d=768\). The projector consists of a linear layer, GELU nonlinearity, a second linear layer, and layer normalization.

The projected pocket vectors are inserted as a prefix before the molecular token embeddings. Given corrupted molecule tokens \(\tilde{x}_t\) at diffusion time \(t\), GenMol-P forms the transformer input
\[
    [p_1,\ldots,p_L,e(\tilde{x}_{t,1}),\ldots,e(\tilde{x}_{t,T})],
\]
where \(e(\cdot)\) is the molecular token embedding. The transformer attends jointly over the pocket prefix and molecular positions, but logits are read only from the molecular positions. Thus, the pocket prefix conditions molecular denoising without being treated as tokens to be generated. The total prefix-plus-molecule length is capped at 256 positions, matching the GenMol positional budget. Training examples that exceed this budget are excluded during preprocessing. This construction keeps the molecular generation mechanism unchanged while allowing the denoising network to condition on the target pocket through continuous structural context.

\subsection{Training}

GenMol-P is supervised-tuned on CrossDocked2020~\citep{francoeur2020three} pocket-ligand pairs. For each complex, the ligand SMILES string is converted to a SAFE string and tokenized with the GenMol tokenizer. The pocket is converted into a residue-level structural prefix by extracting, or deterministically reconstructing, the \(N\), \(C_\alpha\), and \(C\) backbone coordinates for each pocket residue. We discard examples without an assigned train/validation/test split, examples whose ligand cannot be converted to a nonempty SAFE string, malformed pockets with missing backbone information, and examples whose pocket-prefix plus ligand-token length exceeds 256.

The supervised objective is the same masked discrete diffusion objective as GenMol, with the pocket prefix supplied as additional context. For a training pair \((x,\mathcal C)\), we sample a diffusion time \(t\in[\varepsilon_{\mathrm{time}},1]\) with \(\varepsilon_{\mathrm{time}}=10^{-3}\), corrupt the ligand token sequence through the masked discrete diffusion forward process, and train the model to recover the original SAFE tokens from \((\tilde{x}_t,\mathcal C)\). We use antithetic time sampling within minibatches and optimize the globally averaged molecular-token loss. The ESM-IF1 pocket encoder is frozen throughout training. The GenMol backbone and the pocket projector are trainable. The GenMol backbone is initialized from the pretrained GenMol checkpoint, so supervised tuning learns pocket conditioning while retaining the molecular prior learned by the unconditional generator.

The reported GenMol-P checkpoint is trained on 8 H200 GPUs with bf16 precision, per-device batch size 384, and no gradient accumulation, giving a global batch size of 3072. We use AdamW with learning rate \(3\times 10^{-4}\), \(\beta_1=0.9\), \(\beta_2=0.999\), optimizer \(\epsilon=10^{-8}\), and zero weight decay. The learning-rate schedule is constant after 2500 warmup steps, gradients are clipped at norm 1.0, and an exponential moving average with decay 0.9999 is maintained over the trainable parameters. The original GenMol-P model used in the pocket-conditioned experiments is the verified checkpoint at 5,500 supervised optimization steps, and the subsequent coupled-GRPO and coupled-SGRPO runs are initialized from this checkpoint.

\section{Experimental Implementation Details}

\subsection{\emph{De novo} Small-Molecule Design}
\label{app:denovo_small_molecule_details}
\paragraph{Model.}
The \emph{de novo} small-molecule design experiment uses GenMol~\citep{lee2025genmol} as the base generator. GenMol is a masked discrete diffusion language model over SAFE molecular strings~\citep{noutahi2024gotta}. SAFE represents a molecule as an unordered sequence of fragment blocks: a molecule is decomposed into BRICS fragments, each fragment is written as a contiguous string with its attachment points preserved, and fragments are concatenated with separator tokens. This representation is well matched to non-autoregressive denoising because the molecular identity is insensitive to the order in which fragment blocks are listed. As a result, GenMol can model the whole molecular string bidirectionally and fill multiple masked positions in parallel, instead of committing to a left-to-right token order.

GenMol follows the masked discrete diffusion formulation of masked diffusion language models. For a clean SAFE sequence \(x=(x_1,\ldots,x_L)\), the forward process independently corrupts each token by interpolating between the clean token and a mask token \(m\):
\[
q(z_t^l \mid x_l)
=
\mathrm{Cat}
\left(
    z_t^l;\,
    \alpha_t x_l + (1-\alpha_t)m
\right),
\]
where \(z_t^l\) is the noisy token at diffusion time \(t\), and \(\alpha_t\) decreases from \(\alpha_0=1\) to \(\alpha_1=0\). A BERT-style denoiser \(x_\theta(z_t,t)\) predicts the clean token distribution at each position from the partially masked sequence. It is trained with the masked-diffusion negative-ELBO objective, which can be viewed as a time-weighted masked language modeling loss:
\[
\mathcal L_{\mathrm{diff}}
=
\mathbb E_{x,t,z_t}
\left[
    -\sum_{l:z_t^l=m}
    \log x_\theta^l(x_l \mid z_t,t)
\right].
\]
During generation, GenMol starts from a masked SAFE sequence and simulates the reverse process. Unmasked tokens are kept fixed, while each masked position is predicted from \(x_\theta(z_t,t)\) and sampled with temperature \(\tau\). GenMol then confirms only the most confident predictions and leaves the remaining positions masked for later denoising steps. With randomness \(r\), the confidence score for a sampled token \(\hat x_l\) is
\[
c_t^l
=
\log p_\theta^l(\hat x_l\mid z_t,t)
+
r t \epsilon_l,
\qquad
\epsilon_l\sim\mathrm{Gumbel}(0,1).
\]
This confidence-based parallel unmasking gives GenMol a native diffusion sampler rather than an autoregressive action sequence. All post-training methods in this experiment, therefore, use the coupled-GRPO instantiation described in the main text, which evaluates completed molecules under paired diffusion masks while preserving GenMol's native denoising sampler. The original-model baseline is the pretrained GenMol checkpoint without any RL post-training.

\paragraph{Reward definition.}
The rollout-level utility is the same QED--SA reward used in Section~\ref{subsec:denovo_small_molecule}. For a generated molecule \(x\), we compute QED~\citep{bickerton2012quantifying} and the raw synthetic accessibility score \(\mathrm{SA}(x)\)~\citep{ertl2009estimation}. QED is a normalized drug-likeness score that combines several medicinal-chemistry descriptors, including molecular weight, lipophilicity, hydrogen-bond donors and acceptors, polar surface area, rotatable bonds, aromatic rings, and structural alerts, into a single high-is-better value in \([0,1]\). The raw SA score is a heuristic estimate of synthetic difficulty: it rewards molecules composed of common chemical fragments and penalizes structural complexity, such as large rings, stereochemical complexity, and unusual molecular topology. Since lower raw SA indicates easier synthesis, we use the high-is-better transformation \(s_{\mathrm{SA}}(x)=\mathrm{clip}((6-\mathrm{SA}(x))/5,0,1)\). The rollout utility is \(u(x)=0.6\,\mathrm{QED}(x)+0.4\,s_{\mathrm{SA}}(x)\). For SGRPO, the group-level diversity reward is internal diversity over valid molecules, computed as one minus the mean pairwise Tanimoto similarity between Morgan fingerprints. Morgan fingerprints encode circular atom neighborhoods as binary substructure features, and Tanimoto similarity measures the overlap between two such feature sets. Thus, SGRPO optimizes the same molecule-level quality reward as GRPO, but adds an explicit finite-group estimate of sample diversity.

\paragraph{Baselines.}
\camera{We compare SGRPO against five baselines: pretrained GenMol (Original), coupled-GRPO (GRPO), Memory-Assisted GRPO, Diverse Mini-Batch GRPO (DMB-GRPO), and Entropy-Regularized GRPO (ER-GRPO).} Original measures the utility-diversity frontier before RL post-training. Coupled-GRPO~\citep{gong2025diffucoder} uses the rollout utility reward but has no explicit set-level diversity reward. Directly applying GRPO to a masked diffusion generator is nontrivial because a completed sample does not come with the autoregressive factorization normally used to compute token-level policy ratios. Treating the whole completion as fully masked gives a weak likelihood proxy, while independently resampling masks gives high-variance estimates. Coupled-GRPO keeps the GRPO reward and advantage structure but changes how the diffusion policy ratio is estimated.

For a condition \(c\), coupled-GRPO samples a group of completed molecules \(\{x_i\}_{i=1}^G\) and computes a group-relative utility advantage
\[
A_i
=
u(x_i)
-
\frac{1}{G}
\sum_{j=1}^G u(x_j).
\]
The update then uses a PPO-style clipped objective, but the log-probability proxy is computed through coupled diffusion masks. For each completed molecule \(x_i\), coupled-GRPO samples timestep pairs \((t,\hat t)\) with \(t+\hat t=T\), and constructs complementary completion masks \(M_t\) and \(\widehat M_{\hat t}\) such that \(M_t^l+\widehat M_{\hat t}^l=1\) for every completion token \(l\). A compact way to write the corresponding log-probability proxy is
\[
\begin{aligned}
\ell_\theta(x_i;c)
=
\sum_{(t,\hat t)}
\Bigg[
\sum_{l:M_t^l=1}
\log p_\theta(x_i^l \mid c,x_i^{\neg M_t},t)
+
\sum_{l:\widehat M_{\hat t}^l=1}
\log p_\theta(x_i^l \mid c,x_i^{\neg \widehat M_{\hat t}},\hat t)
\Bigg],
\end{aligned}
\]
where \(x_i^{\neg M}\) denotes the visible, unmasked part of the completion under mask \(M\). Thus, each token contributes exactly once across the coupled pair, but is evaluated under a realistic, partially observed context rather than an all-mask context. The policy ratio in the clipped GRPO objective is then formed from this coupled log-probability proxy. Coupled-GRPO therefore preserves the native diffusion denoising interface and reduces the variance of the policy-gradient estimate relative to independently sampled masks. In our baseline, it optimizes only rollout-level utility, without any set-level diversity reward.

Memory-assisted coupled-GRPO, denoted Memory-Assisted GRPO, augments coupled-GRPO with the memory unit of memory-assisted RL~\citep{blaschke2020memory}. We use the compound-similarity version of this memory: it is organized as index--bucket pairs, where each index is a molecule that represents a similarity neighborhood, and each bucket stores the high-scoring generated molecules assigned to that neighborhood. For a newly generated molecule \(x\), the memory is queried only if its utility score exceeds a threshold \(\eta\). The method then compares \(x\) with all indexed molecules using Morgan-fingerprint Tanimoto similarity. If no index has similarity at least \(\gamma\), a new index--bucket pair is created with \(x\), and \(x\) is stored in the new bucket. If \(x\) matches an existing index and the corresponding bucket has not reached its capacity \(C\), \(x\) is added to that bucket and the reward is left unchanged. If the matched bucket is already full, the memory returns zero, and the reward is suppressed before GRPO advantage computation. Thus, Memory-Assisted GRPO encourages diversity indirectly by suppressing reward in historically saturated similarity neighborhoods, rather than by optimizing the internal diversity of the current generated set.

\camera{DMB-GRPO adapts diverse mini-batch selection~\citep{svensson2025diverse} to coupled-GRPO. It generates a candidate pool twice the optimized sample count and selects a diverse subset for the update using a fixed-cardinality determinantal point process (\(k\)-DPP), which favors subsets with a large determinant of their similarity kernel. This promotes diversity through training-sample selection while preserving the utility-based optimization objective.}

\camera{ER-GRPO adds a policy-entropy bonus following \href{https://proceedings.mlr.press/v48/mniha16.html}{Mnih et al.\ (2016)}. For a categorical policy distribution \(p_\theta\) over vocabulary \(\mathcal V\), we use normalized entropy \(H(p_\theta)/\log|\mathcal V|\), where \(H(p_\theta)=-\sum_{v\in\mathcal V}p_\theta(v)\log p_\theta(v)\), with coefficient \(\lambda_{\mathrm{ent}}=0.01\). For GenMol, \(p_\theta\) is the denoiser distribution evaluated under the coupled masking estimator. ER-GRPO reuses policy logits and retains the GRPO rollout count.}

\paragraph{Training details.}
All RL-trained \emph{de novo} small-molecule models start from the same GenMol checkpoint and are trained on 8 H200 GPUs with bf16 precision and seed 42. Training rollouts are sampled with random masking enabled, GenMol randomness \(0.3\), sampling temperature \(1.0\), minimum added length 60, generation batch size 1024, and one training iteration per rollout batch. GRPO uses 512 rollouts per prompt, a per-device training batch size of 1024, learning rate \(5\times 10^{-5}\), Adam \((\beta_1,\beta_2)=(0.9,0.99)\), optimizer \(\epsilon=10^{-8}\), weight decay \(0.1\), maximum gradient norm \(0.2\), KL coefficient \(\beta=0.005\), clipping parameter \(\epsilon_{\mathrm{clip}}=0.5\), cosine learning-rate decay with minimum learning-rate ratio \(0.1\), and warmup ratio \(10^{-4}\). The reference model is synchronized every 64 steps with a mixup coefficient \(0.6\). The reported GRPO model is the 2,000-step checkpoint. Memory-Assisted GRPO uses the same GRPO settings and additionally enables the history-based memory with bucket size 25, score threshold \(0.9\), and similarity cutoff \(0.4\). The reported model is also the 2,000-step checkpoint. SGRPO uses the same optimizer, KL, clipping, scheduler, and training-sampling settings, but keeps the total rollout budget matched by forming 8 same-condition groups with 64 rollouts per group. We set the group-reward weight to \(\lambda=0.5\) and use leave-one-out group credit at temperature \(1.0\). The reported SGRPO model is the 2,000-step checkpoint.

\camera{DMB-GRPO and ER-GRPO use the same initialization, utility reward, optimizer, learning rate, 2,000 training steps, KL coefficient, and other training settings as GRPO. DMB-GRPO matches the number of optimized samples while generating twice as many candidates; ER-GRPO uses the same rollout count as GRPO. All methods use the same evaluation sample counts, decoding settings, and five seeds.}

\paragraph{Evaluation protocol.}
For the main Pareto curve, we evaluate all six methods at the ten paired GenMol randomness--temperature settings in Eq.~\eqref{eq:app_ten_point_sweep}. Each model generates 1,000 molecules at each sweep point. QED, transformed SA, and utility are averaged over valid generated molecules at that sweep point, and diversity is computed over the same valid set using Morgan-fingerprint Tanimoto distance. The frontier metrics in Table~\ref{tab:frontier_metrics} are computed from these \camera{ten} operating points for each method.

\subsection{Pocket-Based Small-Molecule Design}
\label{app:pocket_design_details}
\paragraph{Model.}
The pocket-based small-molecule experiment uses GenMol-P, the pocket-conditioned extension of GenMol described in Appendix~\ref{sec:app_genmol_p}. GenMol-P conditions molecular denoising on a protein pocket by prepending a continuous pocket prefix, obtained from frozen ESM-IF1 pocket embeddings and a trainable projector, before the SAFE-token embeddings. The original GenMol-P model used in this experiment is the supervised CrossDocked2020 checkpoint at 5,500 optimization steps. All RL post-training methods start from this same checkpoint and use the coupled-GRPO family of objectives because the generator remains a discrete diffusion model.

\paragraph{Reward definition.}
The pocket-conditioned utility extends the \emph{de novo} QED--SA reward with a target-dependent docking term. The QED and transformed SA components retain the same interpretation as in the unconditional small-molecule task: QED favors drug-like physicochemical profiles, while \(s_{\mathrm{SA}}\) favors molecules estimated to be easier to synthesize. For ligand \(x\) and pocket \(\mathcal C\), we define \(u(x,\mathcal C)=0.3\,\mathrm{QED}(x)+0.2\,s_{\mathrm{SA}}(x)+0.5\,s_{\mathrm{dock}}(x,\mathcal C)\), where \(s_{\mathrm{SA}}(x)=\mathrm{clip}((6-\mathrm{SA}(x))/5,0,1)\). The docking term is derived from AutoDock Vina~\citep{trott2010autodock}, which searches ligand poses in the pocket and scores them with an empirical approximation to binding affinity. If \(a_{\mathrm{Vina}}(x,\mathcal C)\) is the raw Vina affinity, lower and more negative values indicate stronger predicted binding, so we use the high-is-better transformation \(s_{\mathrm{dock}}(x,\mathcal C)=\mathrm{clip}(-a_{\mathrm{Vina}}(x,\mathcal C)/10,0,1)\). Diversity is condition-specific: for each pocket, we compute internal diversity among the ligands generated for that same pocket using the Morgan-fingerprint Tanimoto distance, and then average this value over pockets. This avoids rewarding diversity across unrelated targets, where chemical dissimilarity may simply reflect different pocket requirements rather than useful within-pocket diversity.

\paragraph{Baselines.}
\camera{We compare SGRPO against the original supervised GenMol-P model (Original), coupled-GRPO (GRPO), DMB-GRPO, and ER-GRPO. All RL methods use the same QED--SA--Vina rollout utility. DMB-GRPO uses the molecular \(k\)-DPP selection described in Appendix~\ref{app:denovo_small_molecule_details}. ER-GRPO uses the same normalized entropy bonus with \(\lambda_{\mathrm{ent}}=0.01\), evaluated on GenMol-P denoiser logits under the coupled masking estimator.} We do not include a memory-assisted GRPO baseline in the main pocket-based comparison. The memory mechanism stores historically generated candidates across optimization steps, but the scientific diversity objective in this task is pocket-conditional: diversity should be encouraged among ligands for the same pocket, not across ligands generated for unrelated pockets. Making a memory-assisted baseline operate at the pocket level would require repeated optimization steps on the same pocket so that the memory becomes meaningful, which would change the data exposure and make the comparison no longer matched.

\paragraph{Training details.}
\camera{All pocket-conditioned RL methods} are trained on 8 H200 GPUs with bf16 precision from the same 5,500-step GenMol-P checkpoint and seed 42. Training rollouts are sampled with random masking enabled, GenMol randomness \(0.3\), sampling temperature \(1.0\), minimum added length 60, generation batch size 384, and one training iteration per rollout batch. Coupled-GRPO uses 192 rollouts per pocket condition, per-device training batch size 384, learning rate \(5\times 10^{-5}\), Adam \((\beta_1,\beta_2)=(0.9,0.99)\), optimizer \(\epsilon=10^{-8}\), weight decay \(0.1\), maximum gradient norm \(0.2\), KL coefficient \(\beta=0.005\), clipping parameter \(\epsilon_{\mathrm{clip}}=0.5\), cosine learning-rate decay with minimum learning-rate ratio \(0.1\), and warmup ratio \(10^{-4}\). The reference model is synchronized every 64 steps with a mixup coefficient \(0.6\). AutoDock Vina rewards are computed with fast search mode, one output pose, docking batch size 384, and a timeout of 1800 seconds. The reported GRPO checkpoint is trained for 1,000 steps. Coupled-SGRPO keeps the total rollout budget matched by forming 8 same-pocket groups with 24 rollouts per group. It uses the same optimizer, KL, clipping, scheduler, docking, and training-sampling settings as GRPO, sets the group-reward weight to \(\lambda=0.9\), and uses leave-one-out group credit at temperature \(1.0\). The reported SGRPO checkpoint is also trained for 1,000 steps.

\camera{DMB-GRPO and ER-GRPO share the GRPO initialization, utility reward, optimizer, learning rate, 1,000 training steps, KL coefficient, and other training settings. DMB-GRPO matches the optimized sample count using twice as many candidate rollouts, while ER-GRPO retains the GRPO rollout count. All methods use the same evaluation sample counts, decoding settings, and five seeds.}

\paragraph{Evaluation protocol.}
The main pocket-based Pareto curve uses pockets from the CrossDocked2020 test set. We evaluate each model under the same paired \mbox{\((\rho,\tau)\)} sweep as the settings in Eq.~\eqref{eq:app_ten_point_sweep}. At each sweep point, each model generates 16 ligands for each of 100 test pockets, giving 1,600 ligand samples per method and sweep point. We scored generated ligands with QED, transformed SA, and AutoDock Vina. Utility metrics are averaged over valid generated ligands, while diversity is computed within each pocket’s 16 ligands and then averaged over the 100 pockets. The frontier metrics in Table~\ref{tab:frontier_metrics} are computed from the resulting \camera{ten} utility-diversity operating points.

\subsection{\emph{De novo} Protein Design}
\label{app:protein_design_details}
\paragraph{Model.}
The \emph{de novo} protein design experiment uses ProGen2-small~\citep{nijkamp2023progen2}, an autoregressive protein language model, as the base generator. ProGen2 models raw protein sequences with a decoder-only Transformer trained for next-token prediction. For an amino-acid sequence \(y=(a_1,\ldots,a_L)\), the model defines a left-to-right sequence distribution
\[
    p_\theta(y)
    =
    \prod_{t=1}^{L+1}
    p_\theta(a_t \mid a_{<t}),
\]
where \(a_{L+1}\) denotes the terminal token. Pretraining minimizes the corresponding negative log-likelihood over large collections of unaligned protein sequences, so generation is performed by repeatedly sampling the next amino-acid token from \(p_\theta(\cdot\mid a_{<t})\) until the terminal token or a maximum length is reached.

Architecturally, ProGen2 is a causal Transformer decoder with rotary positional encodings and a parallelized residual block in which self-attention and the feed-forward network are applied to the same normalized hidden state.
\[
    h^{(m+1)}
    =
    h^{(m)}
    +
    \mathrm{Attn}(\mathrm{LN}(h^{(m)}))
    +
    \mathrm{MLP}(\mathrm{LN}(h^{(m)})).
\]
The ProGen2 family scales this architecture across model sizes. We use the official ProGen2-small checkpoint, which has 151M parameters, 12 layers, 16 attention heads, head dimension 64, and context length 1024. This checkpoint is pretrained on the standard ProGen2 mixture of UniRef90 and BFD30 sequences, giving it broad coverage of natural protein sequence statistics. The original baseline is this checkpoint evaluated without RL post-training. Since ProGen2 provides an explicit token-level likelihood for every generated sequence, SGRPO can be instantiated with standard GRPO rather than the coupled-GRPO estimator required for diffusion generators. All post-training methods generate unconditional amino-acid sequences from the same prompt set and are evaluated at the 100-step checkpoint used in the main comparison.

\paragraph{Reward definition.}
The sequence-level utility combines four normalized protein-design scores, each targeting a different failure mode of unconstrained protein generation. Naturalness is computed from the average per-token log-likelihood under ESM2, so sequences that look more plausible under a large protein language model receive higher scores. Foldability is the ESMFold mean pLDDT score divided by 100. pLDDT is a per-residue confidence estimate, so this term favors sequences whose predicted structures are internally confident. Stability is based on TemBERTure-predicted melting temperature, which serves as a sequence-based proxy for thermal robustness. Developability combines Protein-Sol solubility with simple liability filters, so it penalizes sequences that may be hard to express or handle experimentally. Naturalness and stability are quantile-normalized using calibration sequences, while foldability and developability are already on a \([0,1]\)-compatible scale. The rollout utility is \(u(y)=0.25\,r_{\mathrm{nat}}(y)+0.30\,r_{\mathrm{fold}}(y)+0.20\,r_{\mathrm{stab}}(y)+0.25\,r_{\mathrm{dev}}(y)\). For developability, the underlying score is \(r_{\mathrm{dev}}(y)=0.8\,r_{\mathrm{sol}}(y)+0.2\,r_{\mathrm{liability}}(y)\), where \(r_{\mathrm{sol}}\) is the Protein-Sol score clipped to \([0,1]\) and \(r_{\mathrm{liability}}\) penalizes transmembrane-like hydrophobicity, low sequence complexity, long hydrophobic runs, and cysteine outlier frequency. Sequence diversity is measured as one minus the average normalized Levenshtein similarity over valid generated sequences. The normalized Levenshtein similarity compares two sequences by their edit distance after accounting for sequence length, so lower similarity corresponds to larger sequence-level variation.

\paragraph{Baselines.}
\camera{We compare SGRPO against the original ProGen2-small generator (Original), GRPO, Memory-Assisted GRPO, DMB-GRPO, and ER-GRPO.} GRPO optimizes only the rollout-level protein utility. Memory-Assisted GRPO uses the same rollout utility but adds the history-based diversity memory over generated amino-acid sequences. Sequence similarity in this memory is normalized Levenshtein similarity. The low diversity of Memory-Assisted GRPO in the protein experiment should be interpreted as a consequence of its history-relative novelty mechanism rather than a difference in training budget: all RL methods start from the same ProGen2 checkpoint, use the same rollout utility and matched GRPO training settings, and are evaluated under the same temperature sweep.

\camera{DMB-GRPO adapts diverse mini-batch selection~\citep{svensson2025diverse} to autoregressive GRPO. It generates twice as many candidate sequences as are optimized and uses a task-matched MaxMin selection rule that favors sequences distant from those already selected. ER-GRPO applies the normalized entropy bonus defined in Appendix~\ref{app:denovo_small_molecule_details} to the next-token policy, with \(\lambda_{\mathrm{ent}}=0.01\). It reuses policy logits and retains the GRPO rollout count.}

\paragraph{Training details.}
All ProGen2 post-training runs are trained on 8 H200 GPUs and use the official ProGen2-small checkpoint as both the initial policy and the initial reference model. GRPO uses group size 96, per-device training batch size 192, gradient accumulation 1, maximum generation length 256, top-\(p=0.95\), training sampling temperature \(0.8\), learning rate \(5\times 10^{-5}\), Adam \((\beta_1,\beta_2)=(0.9,0.999)\), optimizer \(\epsilon=10^{-8}\), zero weight decay, maximum gradient norm \(1.0\), KL coefficient \(\beta=0.01\), and clipping parameter \(\epsilon_{\mathrm{clip}}=0.2\). Memory-Assisted GRPO uses the same training settings and enables the history-based diversity memory with bucket size 25, score threshold \(0.6\), and similarity cutoff \(0.6\). In each supergroup, SGRPO keeps the rollout budget matched by forming 8 groups with 12 sequences per group. It uses the same optimizer, generation, KL, and clipping settings as GRPO, sets the group-reward weight to \(\lambda=0.8\), and uses leave-one-out group credit at temperature \(1.0\). For \camera{GRPO, Memory-Assisted GRPO, and SGRPO}, reward calibration uses 1,024 generated sequences. Naturalness, stability, and developability rewards are computed at every step, while the more expensive foldability reward is computed every 4 steps. Each model is trained 100 steps.

\camera{DMB-GRPO and ER-GRPO use the same initialization, rollout utility, optimizer, learning rate, 100 training steps, KL coefficient, and other task-specific settings as GRPO. DMB-GRPO matches the optimized sample count using twice as many candidate rollouts, while ER-GRPO uses the GRPO rollout count. All methods use the same evaluation sample counts, temperatures, and five seeds.}

\paragraph{Evaluation protocol.}
The protein Pareto curve is evaluated with a temperature sweep \(\tau\in\{0.1,0.2,\ldots,1.0,1.1,1.2\}\). At each temperature, each model generates 512 sequences. Naturalness and stability are calibrated once using a fixed set of 256 calibration sequences. The resulting normalization parameters are then held fixed for all models and all temperatures in the sweep.
The plotted utility is the weighted utility defined above, and the frontier metrics in Table~\ref{tab:frontier_metrics} are computed from the 12 temperature-specific operating points for each method.





\input{sec/revision_experiments}

\camera{
\section{Computational Details}
\label{app:cost_details}

\subsection{Computational Cost Under Matched Rollouts}
\label{subsec:computational_cost}

For \(B\) conditions and \(M\) groups of \(K\) candidates per condition, write \(N=MK\). The \(M\) groups partition the \(N\) rollouts rather than multiplying an existing budget by \(M\). Candidate generation, utility evaluation, and policy optimization process the same \(BN\) candidates as GRPO. The additional work is diversity scoring and credit assignment. Let \(T_D(s)\) denote the cost of evaluating the diversity function on \(s\) candidates. Without metric-specific reuse, a black-box set score requires one evaluation on each complete group and \(K\) leave-one-out evaluations:
\begin{equation}
O\!\left(BM[T_D(K)+K T_D(K-1)]\right).
\end{equation}
Group centering, contribution standardization, sign-aware weighting, and reward composition require \(O(BN)\) additional arithmetic.

For pairwise diversity, let \(T_m=\sum_{i<j}d(x_{m,i},x_{m,j})\) and \(t_{m,i}=\sum_{j\ne i}d(x_{m,i},x_{m,j})\). Accumulating the total and all row sums in one pass gives
\begin{equation}
D(G_m)=\frac{2T_m}{K(K-1)},
\qquad
D(G_m\setminus\{x_{m,i}\})
=\frac{2(T_m-t_{m,i})}{(K-1)(K-2)}
\quad (K>2).
\label{eq:app_pairwise_reuse}
\end{equation}
Thus all \(K\) leave-one-out contributions follow without additional similarity evaluations. The pair count is \(BM\binom K2=BN(K-1)/2\). Writing \(C_{\mathrm{sim}}\) for the cost of one similarity evaluation, the additional similarity cost is \(O(BNK C_{\mathrm{sim}})\), followed by linear post-processing. At \(K=2\), a convention for singleton diversity is required. The smallest nontrivial SGRPO group size in our reported ablation is \(K=4\), while \(K=1\) denotes the separate utility-only GRPO setting.

Constructing a metric representation for each candidate is an additional per-candidate cost. Once fixed-length Morgan fingerprints are available, a Tanimoto comparison operates on fixed-size bit vectors. For two protein sequences of length at most \(L\), Levenshtein dynamic programming has \(O(L^2)\) worst-case cost. The pairwise reward calculation can therefore be more expensive for longer proteins. Reusing row sums removes repeated scoring of leave-one-out subsets but does not remove the metric's underlying pair cost.

\subsection{End-to-End Training Efficiency}

We profile \emph{de novo} small-molecule training on eight H200 GPUs with 8,192 optimized samples per step. Table~\ref{tab:training_efficiency} reports mean step time and throughput over the full 2,000-step runs. SGRPO takes 18.96\,s per step versus 15.10\,s for GRPO, a 25.6\% increase.

The profiling in Table~\ref{tab:training_efficiency} measures complete training steps, not an isolated diversity-scoring kernel. It shows the practical overhead under the reported implementation and hardware, but it does not identify how much time is attributable to each component. In particular, DMB-GRPO matches optimized sample counts but draws twice as many candidates, so its timing comparison should not be interpreted as matched candidate-generation cost.
}

\begin{table}[htbp]
\centering
\caption{\camera{End-to-end efficiency on \emph{de novo} small-molecule design, averaged over 2,000 training steps on eight H200 GPUs with 8,192 optimized samples per step. DMB-GRPO generates twice as many candidates to select the same number for optimization.}}
\label{tab:training_efficiency}
\small
\camera{\begin{tabular}{lrr}
\toprule
Method & Mean step time (s) & Optimized samples/s\\
\midrule
GRPO & 15.10 & 542.50\\
Memory-Assisted GRPO & 20.12 & 407.22\\
DMB-GRPO & 27.27 & 300.45\\
ER-GRPO & 17.81 & 459.85\\
SGRPO & 18.96 & 432.07\\
\bottomrule
\end{tabular}}
\end{table}

\section{\camera{Limitations and Future Work}}
\label{app:limitations}

SGRPO trades an additional rollout structure for a direct set-level diversity signal. Unlike diversity-agnostic RL, SGRPO must sample groups of candidates under the same conditions in order to estimate and compare group diversity. Our analysis shows that SGRPO remains effective with a near-minimal group size, but the method can still require more rollouts than objectives that score each candidate independently. Improving the efficiency of finite-group diversity estimators, or reusing rollouts across compatible diversity computations, is therefore an important direction for making SGRPO cheaper at larger scales.

SGRPO also adds the cost of computing the diversity reward and the leave-one-out diversity contribution. For the pairwise-similarity diversity objectives used in our experiments, this cost is dominated by constructing the within-group similarity matrix. Let a supergroup contain \(M\) same-condition groups, each with \(K\) rollouts, and let \(C_{\mathrm{sim}}\) denote the cost of one similarity evaluation. Computing all within-group pairs requires \(O(MK^2 C_{\mathrm{sim}})\) work, or \(M K(K-1)/2\) pair evaluations up to constants. The same similarity matrix is then reused to compute both the group diversity reward and all leave-one-out contributions, so leave-one-out credit assignment does not require a second pass over the pairwise similarities. Our implementation follows this reuse pattern to avoid redundant similarity computation.

This overhead should be interpreted relative to other diversity-aware training mechanisms. Under the matched-rollout protocol used in our experiments, Memory-Assisted GRPO evaluates the same \(MK\) generated candidates but compares high-scoring candidates against a memory with \(B\) index--bucket pairs. Its memory lookup therefore scales as \(O(MKB C_{\mathrm{sim}})\) in the number of stored similarity neighborhoods. Once \(B\) is group-scale or larger, this lookup can exceed the pairwise group-diversity computation. In the \emph{de novo} small-molecule run, for example, SGRPO uses \(K=64\) and \(M=8\), while the memory already contains more than 300 index--bucket pairs after the first training step. 

The diversity-computation overhead is nevertheless real. For pairwise-similarity diversity, beyond the similarity-matrix reuse used here, future implementations could cache pairwise similarities across repeated candidates or cache reusable metric-specific computations such as fingerprints, embeddings, or nearest-neighbor structures. For diversity notions that are not naturally pairwise-similarity based, the appropriate efficiency strategy may be different: GPU-batched diversity computation, differentiable or learned proxy metrics, or approximate set summaries may be needed to keep set-level rewards practical at larger rollout scales.

Our experiments evaluate the utility-diversity trade-off with respect to the scalar utility specified for each task. These utilities aggregate several design-relevant components, such as drug-likeness, synthesizability, docking affinity, foldability, stability, and developability, using task-specific weights. This scalarization is the shared interface through which a design preference is provided to reward-based post-training: the method is given a utility function and a diversity metric, and the evaluation asks whether the attainable trade-off between them improves. Questions about how to choose, calibrate, or audit the internal utility components are important reward-design questions shared by all reward-based post-training methods. They are complementary to the contribution studied here: given a user-specified utility and diversity metric, SGRPO directly improves the attainable trade-off between them. Although the scope of this work treats utility and diversity as user-specified evaluation axes, a natural extension is to expose the utility components themselves as additional controllable axes, allowing future SGRPO variants to study diversity jointly with more fine-grained within-utility trade-offs.

SGRPO is intentionally decoupled from a specific generator, task, utility objective, or diversity metric. This makes it easy to instantiate through different GRPO-style optimizers, but it also leaves room for more specialized designs. Future work could use the same supergroup-relative framework with task-specific diversity notions, condition-aware chemical series constraints, structure-aware protein diversity metrics, or adaptive group construction rules that better match the scientific objective of a particular design campaign.

\camera{\paragraph{Diversity definitions and guarantees.}
Replacing the set-level diversity function changes the property encouraged by the reward, even when the optimization procedure remains unchanged. Exact population and partition alignment is established here for normalized fixed-order U-statistics. The scaffold experiment supplies empirical evidence for a non-pairwise metric. The candidate-wise redistribution also introduces the credit-shaping term in Eq.~\eqref{eq:app_gradient_decomposition}. However, reward conservation alone does not remove this term or provide a general convergence guarantee.}

\camera{\paragraph{Computational proxies and biological validation.}
Our evaluations measure the attainable trade-off between specified computational utility and diversity scores. Whether a utility proxy predicts biological activity or functional relevance is a separate validation question shared by reward-based design methods. Improvements in these computational metrics do not establish experimental activity, efficacy, or functional diversity. Experimental validation of selected candidates, and assessment of the correspondence between each utility proxy and measured outcomes, are important future work alongside more task-specific SGRPO instantiations.}

%% file: sec/revision_objective.tex
\camera{
\section{Reward Conservation and the Optimization Objective}
\label{app:redistribution_objective}

\subsection{Conservation Under Redistribution}

\begin{proposition}[Conservation of the mean group reward]
\label{prop:reward_conservation}
For the weights in Eq.~\eqref{eq:redistribution_weights} and redistributed rewards in Eq.~\eqref{eq:redistributed_group_reward},
\begin{equation}
\frac{1}{K}\sum_{i=1}^K\widetilde R_{m,i}=R_m=D(G_m),
\qquad
\sum_{i=1}^K\widetilde R_{m,i}=K R_m.
\label{eq:app_reward_conservation}
\end{equation}
\end{proposition}

\begin{proof}
Both sign-aware weight vectors sum to \(K\). Averaging Eq.~\eqref{eq:redistributed_group_reward} within a group gives
\begin{align}
\frac{1}{K}\sum_i\widetilde R_{m,i}
&=R_m+[A_m^{\mathrm{grp}}]_+
\left(\frac{1}{K}\sum_i w^+_{m,i}-1\right)
-[-A_m^{\mathrm{grp}}]_+
\left(\frac{1}{K}\sum_i w^-_{m,i}-1\right)
=R_m.
\end{align}
Thus redistribution preserves the mean group reward; summing the rollout rewards yields \(K\) times that reward.
\end{proof}

\subsection{Population Target and Credit-Shaped Update}

Let the expected \(K\)-sample diversity under the policy be
\begin{equation}
V_K(\theta)=
\mathbb E_{\mathcal C\sim p(\mathcal C),\,G\sim\pi_\theta^K(\cdot\mid\mathcal C)}[D(G)],
\qquad
s_{m,i}=\nabla_\theta\log\pi_\theta(x_{m,i}\mid\mathcal C).
\label{eq:app_expected_set_diversity}
\end{equation}
Here \(D\) has no explicit parameter dependence. We assume independent policy samples conditional on \(\mathcal C\), \(M\geq2\), and the regularity needed to exchange differentiation and expectation, including parameter-independent support and integrable score-weighted rewards. These statements concern exact likelihood scores; the diffusion implementation uses the coupled likelihood estimator described in the experimental implementation details.

Conditional on \(\mathcal C\), rewards from other groups are independent of the candidates in \(G_m\), and \(\mathbb E[s_{m,i}\mid\mathcal C]=0\). The other-group mean is therefore a valid baseline, giving
\begin{equation}
\mathbb E\!\left[
\frac{1}{M}\sum_{m=1}^M A_m^{\mathrm{grp}}
\sum_{i=1}^K s_{m,i}
\right]
=\nabla_\theta V_K(\theta).
\label{eq:app_group_gradient}
\end{equation}
To see this, the score identity gives \(\mathbb E[R_m\sum_i s_{m,i}]=\nabla_\theta V_K(\theta)\), while the cross-group baseline term has expectation zero.

To characterize the final update, write
\begin{equation}
\delta_{m,i}
=A_m^{\mathrm{grp}}(w_{m,i}^{\sigma_m}-1),
\qquad
\widetilde R_{m,i}=R_m+\delta_{m,i},
\qquad
w_{m,i}^{\sigma_m}=
\begin{cases}
w^+_{m,i},& A_m^{\mathrm{grp}}\geq0,\\
w^-_{m,i},& A_m^{\mathrm{grp}}<0.
\end{cases}
\end{equation}
Conservation implies \(\sum_i\delta_{m,i}=0\), so the mean redistributed diversity over the supergroup is \(\bar R\). The diversity-only part of Eq.~\eqref{eq:supergroup_advantage}, averaged over rollouts and multiplied by their scores, is
\begin{align}
\widehat g_{\mathrm{div}}
&=\frac{1}{MK}\sum_{m,i}
\frac{MK}{MK-1}(\widetilde R_{m,i}-\bar R)s_{m,i}\\
&=\frac{1}{MK-1}\sum_{m,i}(R_m-\bar R)s_{m,i}
+\frac{1}{MK-1}\sum_{m,i}\delta_{m,i}s_{m,i}.
\label{eq:app_gradient_sample_decomposition}
\end{align}
Using \(R_m-\bar R=(M-1)A_m^{\mathrm{grp}}/M\) and Eq.~\eqref{eq:app_group_gradient} yields
\begin{equation}
\mathbb E[\widehat g_{\mathrm{div}}]
=\frac{M-1}{MK-1}\nabla_\theta V_K(\theta)
+B_{\tau_c}(\theta),
\qquad
B_{\tau_c}(\theta)
=\frac{1}{MK-1}\mathbb E\!\left[
\sum_{m,i}A_m^{\mathrm{grp}}(w_{m,i}^{\sigma_m}-1)s_{m,i}
\right].
\label{eq:app_gradient_decomposition}
\end{equation}
The first term is a positive multiple of the exact expected-set-diversity gradient. The second is the candidate-level credit-shaping term: it emphasizes high-contribution candidates in above-average-diversity groups and low-contribution candidates in below-average-diversity groups with the appropriate sign. This term need not vanish, so the complete redistributed update is not generally unbiased for \(\nabla_\theta V_K(\theta)\). With uniform redistribution weights, it vanishes and only the scaled exact gradient component remains.

\subsection{The Local Policy-Optimization Surrogate}

At each iteration, fix the behavior policy \(\mu=\pi_{\theta_{\mathrm{old}}}\), sample the supergroup \(\mathcal S\sim\mu^{MK}(\cdot\mid\mathcal C)\), and compute the rewards and advantages using Eqs.~\eqref{eq:redistributed_group_reward}--\eqref{eq:supergroup_advantage}. These quantities are held fixed during the policy update. The maximization form of Eq.~\eqref{eq:sgrpo_objective} is
\begin{align}
\mathcal J_{\mathrm{local}}(\theta;\mu)
={}&\mathbb E_{\mathcal C,\mathcal S}
\left[\frac{1}{MK}\sum_{m,i}
\min\!\left\{
\rho_{m,i}A_{m,i},
\operatorname{clip}(\rho_{m,i},1-\epsilon,1+\epsilon)A_{m,i}
\right\}\right]\nonumber\\
&-\beta\,\mathbb E_{\mathcal C}\!
\left[\mathrm{KL}\!\left(\pi_\theta(\cdot\mid\mathcal C)
\Vert\pi_{\mathrm{ref}}(\cdot\mid\mathcal C)\right)\right],
\qquad
\rho_{m,i}=\frac{\pi_\theta(x_{m,i}\mid\mathcal C)}
{\mu(x_{m,i}\mid\mathcal C)}.
\label{eq:app_local_surrogate}
\end{align}
At \(\theta=\theta_{\mathrm{old}}\), the diversity component of the surrogate gradient is \(\lambda\) times Eq.~\eqref{eq:app_gradient_decomposition}. Away from the behavior policy, optimization follows this clipped local surrogate. Thus the population diversity target, the credit-shaped local objective, and the estimator used by a particular GRPO-family implementation are distinct.
}

%% file: sec/revision_experiments.tex
\camera{
\section{Extended Evaluation}
\label{app:extended_evaluation}

\subsection{Evaluation Protocol and Metric Interpretation}
\label{app:expanded_protocol}

Both small-molecule tasks use ten paired randomness--temperature settings:
\begin{equation}
\begin{aligned}
(\rho,\tau)\in\{&
(0.1,0.5),(0.2,0.65),(0.3,0.8),(0.4,0.95),(0.5,1.1),\\
&(0.6,1.25),(0.7,1.4),(0.8,1.55),(0.9,1.7),(1.0,2.0)\}.
\end{aligned}
\label{eq:app_ten_point_sweep}
\end{equation}
The protein task retains its twelve temperatures. Each comparison uses five independent generation sweeps of the trained checkpoints, with matched generation budgets and seeds across methods; these are evaluation runs, not five independently trained models. Every frontier metric is computed separately for each run before reporting its mean and 95\% confidence interval. For HV, the reference coordinates are the minimum utility and diversity across all methods and operating points in the corresponding comparison panel for that run. Thus HV values are comparable within a panel, but can change when the methods included in the panel change. Separate ablation and scaffold-diversity tables use their own stated comparison panels.

The three indicators summarize different aspects of a frontier. HV measures geometric coverage, rewarding both outward movement and a broad range of trade-offs. DIP measures only the closest operating point to the ideal and ignores the rest of the frontier. R2 averages the best weighted shortfall over preference directions, assessing whether different preferences have suitable operating points. A single strong compromise can yield a good DIP without broad coverage; a wide frontier can have a high HV without the best preference-averaged shortfall. In Table~\ref{tab:frontier_metrics}, for example, DMB-GRPO has the highest mean small-molecule HV, whereas SGRPO has the lowest mean DIP and R2.

We describe numerical rankings separately from statistical evidence. Overlapping 95\% confidence intervals do not establish a significant difference, and we describe such results as competitive.

\subsection{Larger Within-Pocket Evaluation Sets}
\label{app:pocket_100}

To test sensitivity to the finite evaluation set, we increase the generation budget from 16 to 100 ligands for each of 100 held-out CrossDocked2020 pockets. At every one of the ten decoding settings, each model therefore generates 10,000 ligands per run. Utility is aggregated over valid ligands, while Morgan-fingerprint internal diversity is computed within each pocket and then averaged across pockets. Table~\ref{tab:pocket_100} uses a common per-run HV reference point determined by its three models. SGRPO retains the best mean HV, DIP, and R2. 
}

\begin{table}[htbp]
\centering
\caption{\camera{Pocket-based evaluation with 100 ligands per pocket, ten decoding settings, and five independent sweeps. Entries are mean \(\pm\) 95\% CI; the HV reference is shared only by the three methods in this table.}}
\label{tab:pocket_100}
\small
\camera{\begin{tabular}{lccc}
\toprule
Method & HV \(\uparrow\) & DIP \(\downarrow\) & R2 \(\downarrow\)\\
\midrule
Original & \metricci{0.0261}{0.0012} & \metricci{0.4646}{0.0023} & \metricci{0.2395}{0.0007}\\
GRPO & \metricci{0.0094}{0.0005} & \metricci{0.7421}{0.0069} & \metricci{0.3820}{0.0036}\\
SGRPO & \bestmetricci{0.0632}{0.0023} & \bestmetricci{0.3805}{0.0007} & \bestmetricci{0.1810}{0.0004}\\
\bottomrule
\end{tabular}}
\end{table}

\camera{
\subsection{Raw Diversity Contributions and Alternative Centering}
\label{app:credit_ablations}

\paragraph{Raw-LOO diversity reward.}
We compare SGRPO with a baseline that directly uses the raw leave-one-out contribution \(c_{m,i}=D(G_m)-D(G_m\setminus\{x_{m,i}\})\) as its diversity reward, without the supergroup comparison and sign-aware redistribution. For normalized pairwise diversity with \(K>2\), let \(q_i=(K-1)^{-1}\sum_{j\neq i}d(x_i,x_j)\). Direct calculation gives \(c_i=2(q_i-D_K(G))/(K-2)\), making its relation to mean distance from the other candidates explicit. The experiment uses the same ten-point, five-run \emph{de novo} small-molecule protocol, with 1,000 molecules per setting. Table~\ref{tab:raw_loo} shows lower DIP and R2 and higher HV for SGRPO, supporting the empirical value of the group comparison and redistribution beyond using raw contributions alone.
}

\begin{table}[htbp]
\centering
\caption{\camera{Raw-LOO comparison on \emph{de novo} small-molecule design. Mean \(\pm\) 95\% CI over five sweeps; the HV reference is determined jointly from these two methods for each run.}}
\label{tab:raw_loo}
\small
\camera{\begin{tabular}{lccc}
\toprule
Method & HV \(\uparrow\) & DIP \(\downarrow\) & R2 \(\downarrow\)\\
\midrule
Raw-LOO & \metricci{0.0372}{0.0013} & \metricci{0.2680}{0.0021} & \metricci{0.1064}{0.0008}\\
SGRPO & \bestmetricci{0.0410}{0.0015} & \bestmetricci{0.2551}{0.0020} & \bestmetricci{0.0955}{0.0004}\\
\bottomrule
\end{tabular}}
\end{table}

\camera{
\paragraph{Relative-baseline choice.}
Leave-one-out mean centering without subsequent standard-deviation normalization follows the RLOO formulation of \href{https://aclanthology.org/2024.acl-long.662/}{Ahmadian et al.\ (2024)} and is a design choice rather than a core contribution. It is distinct from computing a candidate's leave-one-out diversity contribution: the former centers rewards against other samples, whereas the latter measures how removing a candidate changes a set score. For \(L\) rewards \(a_1,\ldots,a_L\), leave-one-out centering satisfies \(a_i-\frac{1}{L-1}\sum_{j\ne i}a_j=\frac{L}{L-1}(a_i-\bar a)\), so it differs from mean centering by a fixed positive scale before further transformations.

We compare the default SGRPO centering with mean subtraction and mean subtraction followed by standard-deviation normalization. The three variants use the same ten-point, five-run \emph{de novo} small-molecule evaluation. Table~\ref{tab:centering_ablation} shows broadly comparable performance. Mean centering is slightly better numerically across the three metrics, while Mean + Std has similar HV and R2 to SGRPO but a larger DIP. These results support mean centering as a competitive alternative rather than establishing an advantage for the default centering rule.
}

\begin{table}[htbp]
\centering
\caption{\camera{Centering ablation on \emph{de novo} small-molecule design. Mean \(\pm\) 95\% CI over five sweeps, with the HV reference determined jointly from these three variants per run.}}
\label{tab:centering_ablation}
\small
\camera{\begin{tabular}{lccc}
\toprule
Centering & HV \(\uparrow\) & DIP \(\downarrow\) & R2 \(\downarrow\)\\
\midrule
SGRPO (leave-one-out) & \metricci{0.0407}{0.0019} & \metricci{0.2551}{0.0020} & \metricci{0.0955}{0.0004}\\
Mean & \bestmetricci{0.0415}{0.0018} & \bestmetricci{0.2534}{0.0037} & \bestmetricci{0.0942}{0.0005}\\
Mean + Std & \metricci{0.0411}{0.0019} & \metricci{0.2666}{0.0024} & \metricci{0.0954}{0.0012}\\
\bottomrule
\end{tabular}}
\end{table}

\camera{
\subsection{Transfer to Scaffold Diversity}
\label{app:scaffold_transfer}

We evaluate whether diversity gains transfer to a metric not used for training. Define relative Bemis--Murcko scaffold diversity as
\begin{equation}
D_{\mathrm{scaf}}(S)=
\begin{cases}
\displaystyle\frac{|\{\mathrm{BM}(x):x\in S_{\mathrm{valid}}\}|}{|S_{\mathrm{valid}}|},
&|S_{\mathrm{valid}}|>0,\\[4pt]
0,&|S_{\mathrm{valid}}|=0,
\end{cases}
\label{eq:app_scaffold_diversity}
\end{equation}
where \(S_{\mathrm{valid}}\) retains sample multiplicity and \(\mathrm{BM}(x)\) is the canonical Bemis--Murcko scaffold of \(x\). The numerator counts distinct scaffolds; the denominator counts all valid generated molecules.

For pocket-based design, we reuse the saved generations and utility scores from the ten-point evaluation of Original, GRPO, and Morgan-diversity-trained SGRPO. Each run contains 16 ligands for each of 100 pockets per setting. Scaffold diversity is computed within each pocket and then averaged over pockets. No model is retrained for this comparison. Table~\ref{tab:scaffold_transfer} shows that SGRPO retains the best frontier metrics under this alternative evaluation metric, providing evidence of transfer beyond the pairwise training score.
}

\begin{table}[htbp]
\centering
\caption{\camera{Utility--scaffold-diversity evaluation for pocket-based design, without scaffold-diversity training. Mean \(\pm\) 95\% CI over five sweeps; each run uses a reference point shared by these three methods.}}
\label{tab:scaffold_transfer}
\small
\camera{\begin{tabular}{lccc}
\toprule
Method & HV \(\uparrow\) & DIP \(\downarrow\) & R2 \(\downarrow\)\\
\midrule
Original & \metricci{0.0252}{0.0031} & \metricci{0.4988}{0.0052} & \metricci{0.2586}{0.0033}\\
GRPO & \metricci{0.0060}{0.0013} & \metricci{0.7976}{0.0111} & \metricci{0.4093}{0.0056}\\
SGRPO & \bestmetricci{0.0522}{0.0031} & \bestmetricci{0.4393}{0.0054} & \bestmetricci{0.2269}{0.0015}\\
\bottomrule
\end{tabular}}
\end{table}

\camera{
\subsection{Training with a Non-Pairwise Diversity Reward}
\label{app:scaffold_training}

To test compatibility with a different diversity definition while holding the task and generator fixed, we train SGRPO on \emph{de novo} small-molecule design using \(D_{\mathrm{scaf}}\) in Eq.~\eqref{eq:app_scaffold_diversity}. The optimization procedure is unchanged: group scores and leave-one-out contributions are obtained from this set function. Evaluation uses the utility--scaffold-diversity frontier, ten decoding settings, 1,000 molecules per model and setting, and five independent generation runs.

We compare with Original, GRPO, and ER-GRPO. The implemented Memory-Assisted GRPO and DMB-GRPO baselines rely on pairwise similarities for memory matching and mini-batch selection, respectively. Relative scaffold diversity cannot replace those pairwise quantities directly without redesigning the baselines, so they are omitted from this training comparison. This does not preclude other scaffold-aware adaptations of either method.

Scaffold-trained SGRPO has the best mean values on all three metrics in Table~\ref{tab:scaffold_training}; some confidence intervals overlap. The result supports practical compatibility with a non-pairwise coverage statistic. The fixed-order U-statistic guarantee in Appendix~\ref{app:sgrpo_theory} does not apply to this score.
}

\begin{table}[htbp]
\centering
\caption{\camera{\emph{De novo} small-molecule design with scaffold diversity as the SGRPO training reward and evaluation metric. Mean \(\pm\) 95\% CI over five sweeps, using a reference point shared by the four methods in each run.}}
\label{tab:scaffold_training}
\small
\camera{\begin{tabular}{lccc}
\toprule
Method & HV \(\uparrow\) & DIP \(\downarrow\) & R2 \(\downarrow\)\\
\midrule
Original & \metricci{0.1133}{0.0026} & \metricci{0.2818}{0.0027} & \metricci{0.1125}{0.0027}\\
GRPO & \metricci{0.1287}{0.0032} & \metricci{0.2767}{0.0050} & \metricci{0.1098}{0.0020}\\
ER-GRPO & \metricci{0.1292}{0.0039} & \metricci{0.2667}{0.0024} & \metricci{0.1059}{0.0022}\\
SGRPO (scaffold) & \bestmetricci{0.1353}{0.0042} & \bestmetricci{0.2511}{0.0050} & \bestmetricci{0.0999}{0.0017}\\
\bottomrule
\end{tabular}}
\end{table}

\camera{
\subsection{Motif-Conditioned Molecular Generation}
\label{app:motif_extension}

\paragraph{Setup.}
GenMol motif extension provides a second conditional generation task. The condition is a discrete molecular motif with a single attachment point, and the generated molecule must retain that motif. Training uses 5,543 single-attachment motifs derived from CReM ChEMBL22, excluding the ten held-out motifs in the official GenMol evaluation benchmark. All RL methods start from the same initialization and use 1,000 training steps with matched optimized sample counts. We compare Original, GRPO, DMB-GRPO, ER-GRPO, and SGRPO with \(\lambda=0.5\); the RL methods use coupled-GRPO because the generator is a masked diffusion model.

Each method generates 100 molecules per test motif at each of the ten settings in Eq.~\eqref{eq:app_ten_point_sweep}, with five independent sweeps. Utility assesses full-molecule quality subject to molecular validity and retention of the conditioning motif. Morgan internal diversity is evaluated within each motif condition and then averaged across the ten motifs. This tests a discrete conditioning interface in addition to the three-dimensional pocket interface.

\paragraph{Results.}
SGRPO has the highest mean HV, while ER-GRPO has the lowest mean DIP and R2 (Table~\ref{tab:motif_extension}). SGRPO is competitive with GRPO and ER-GRPO, with overlapping confidence intervals across all three indicators. These results support its applicability to motif-conditioned generation without establishing a significant advantage over those two baselines.
}

\begin{table}[htbp]
\centering
\caption{\camera{Motif-extension frontier metrics. Mean \(\pm\) 95\% CI over five sweeps, with 100 generations per motif and setting. The five methods share a per-run HV reference point.}}
\label{tab:motif_extension}
\small
\camera{\begin{tabular}{lccc}
\toprule
Method & HV \(\uparrow\) & DIP \(\downarrow\) & R2 \(\downarrow\)\\
\midrule
Original & \metricci{0.0036}{0.0013} & \metricci{0.8324}{0.0084} & \metricci{0.4369}{0.0051}\\
GRPO & \metricci{0.0086}{0.0014} & \metricci{0.7795}{0.0073} & \metricci{0.4093}{0.0025}\\
DMB-GRPO & \metricci{0.0054}{0.0007} & \metricci{0.8135}{0.0082} & \metricci{0.4268}{0.0064}\\
ER-GRPO & \metricci{0.0093}{0.0016} & \bestmetricci{0.7756}{0.0104} & \bestmetricci{0.4051}{0.0065}\\
SGRPO & \bestmetricci{0.0096}{0.0018} & \metricci{0.7760}{0.0099} & \metricci{0.4066}{0.0042}\\
\bottomrule
\end{tabular}}
\end{table}

%% file: checklist.tex
\section*{NeurIPS Paper Checklist}

The checklist is designed to encourage best practices for responsible machine learning research, addressing issues of reproducibility, transparency, research ethics, and societal impact. Do not remove the checklist: {\bf The papers not including the checklist will be desk rejected.} The checklist should follow the references and follow the (optional) supplemental material.  The checklist does NOT count towards the page
limit. 

Please read the checklist guidelines carefully for information on how to answer these questions. For each question in the checklist:
\begin{itemize}
    \item You should answer \answerYes{}, \answerNo{}, or \answerNA{}.
    \item \answerNA{} means either that the question is Not Applicable for that particular paper or the relevant information is Not Available.
    \item Please provide a short (1--2 sentence) justification right after your answer (even for \answerNA). 
\end{itemize}

{\bf The checklist answers are an integral part of your paper submission.} They are visible to the reviewers, area chairs, senior area chairs, and ethics reviewers. You will also be asked to include it (after eventual revisions) with the final version of your paper, and its final version will be published with the paper.

The reviewers of your paper will be asked to use the checklist as one of the factors in their evaluation. While \answerYes{} is generally preferable to \answerNo{}, it is perfectly acceptable to answer \answerNo{} provided a proper justification is given (e.g., error bars are not reported because it would be too computationally expensive'' or ``we were unable to find the license for the dataset we used''). In general, answering \answerNo{} or \answerNA{} is not grounds for rejection. While the questions are phrased in a binary way, we acknowledge that the true answer is often more nuanced, so please just use your best judgment and write a justification to elaborate. All supporting evidence can appear either in the main paper or the supplemental material, provided in appendix. If you answer \answerYes{} to a question, in the justification please point to the section(s) where related material for the question can be found.

IMPORTANT, please:
\begin{itemize}
    \item {\bf Delete this instruction block, but keep the section heading ``NeurIPS Paper Checklist"},
    \item  {\bf Keep the checklist subsection headings, questions/answers and guidelines below.}
    \item {\bf Do not modify the questions and only use the provided macros for your answers}.
\end{itemize}


\begin{enumerate}

\item {\bf Claims}
    \item[] Question: Do the main claims made in the abstract and introduction accurately reflect the paper's contributions and scope?
    \item[] Answer: \answerYes{} 
    \item[] Justification: The main claims are supported by the proposed framework in the method and the experimental results on representative datasets.

\item {\bf Limitations}
    \item[] Question: Does the paper discuss the limitations of the work performed by the authors?
    \item[] Answer: \answerYes{}  
    \item[] Justification: Include limitations in the appendix. 

\item {\bf Theory assumptions and proofs}
    \item[] Question: For each theoretical result, does the paper provide the full set of assumptions and a complete (and correct) proof?
    \item[] Answer: \answerYes{} 
    \item[] Justification: Include proof in the appendix.

    \item {\bf Experimental result reproducibility}
    \item[] Question: Does the paper fully disclose all the information needed to reproduce the main experimental results of the paper to the extent that it affects the main claims and/or conclusions of the paper (regardless of whether the code and data are provided or not)?
    \item[] Answer: \answerYes{} 
    \item[] Justification: We describe the datasets, preprocessing protocol, evaluation splits, baselines, metrics, and implementation details needed to reproduce the main experiment results.

\item {\bf Open access to data and code}
    \item[] Question: Does the paper provide open access to the data and code, with sufficient instructions to faithfully reproduce the main experimental results, as described in supplemental material?
    \item[] Answer: \answerYes{} 
    \item[] Justification: The code and data for reproducing the main experimental results are publicly available at \url{https://anonymous.4open.science/r/SGRPO/README.md}.

\item {\bf Experimental setting/details}
    \item[] Question: Does the paper specify all the training and test details (e.g., data splits, hyperparameters, how they were chosen, type of optimizer) necessary to understand the results?
    \item[] Answer: \answerYes{} 
    \item[] Justification: The paper specifies the datasets, train-test splits, evaluation metrics, baselines, model settings, and key hyperparameters used in the experiments.

\item {\bf Experiment statistical significance}
    \item[] Question: Does the paper report error bars suitably and correctly defined or other appropriate information about the statistical significance of the experiments?
    \item[] Answer: \answerYes{} 
    \item[] Justification: The paper reports error bars suitably and correctly defined about the statistical significance of the experiments.

\item {\bf Experiments compute resources}
    \item[] Question: For each experiment, does the paper provide sufficient information on the computer resources (type of compute workers, memory, time of execution) needed to reproduce the experiments?
    \item[] Answer: \answerYes{} 
    \item[] Justification: We provide the details of training in the appendix.
    
\item {\bf Code of ethics}
    \item[] Question: Does the research conducted in the paper conform, in every respect, with the NeurIPS Code of Ethics \url{https://neurips.cc/public/EthicsGuidelines}?
    \item[] Answer: \answerYes{} 
    \item[] Justification: The research follows the NeurIPS Code of Ethics and uses publicly available biological datasets and curated knowledge resources without involving private or sensitive human-subject data.

\item {\bf Broader impacts}
    \item[] Question: Does the paper discuss both potential positive societal impacts and negative societal impacts of the work performed?
    \item[] Answer: \answerYes{} 
    \item[] Justification:  We discuss the potential positive impact of the efficient protein design framework in the introduction motivation part.
    
\item {\bf Safeguards}
    \item[] Question: Does the paper describe safeguards that have been put in place for responsible release of data or models that have a high risk for misuse (e.g., pre-trained language models, image generators, or scraped datasets)?
    \item[] Answer: \answerNA{} 
    \item[] Justification: This work does not release high-risk models, scraped datasets, or other assets with substantial misuse risks.

\item {\bf Licenses for existing assets}
    \item[] Question: Are the creators or original owners of assets (e.g., code, data, models), used in the paper, properly credited and are the license and terms of use explicitly mentioned and properly respected?
    \item[] Answer: \answerYes{} 
    \item[] Justification: We cite the sources of the datasets in the experimental setups part.

\item {\bf New assets}
    \item[] Question: Are new assets introduced in the paper well documented and is the documentation provided alongside the assets?
    \item[] Answer: \answerNA{} 
    \item[] Justification: The paper does not introduce a new public dataset.

\item {\bf Crowdsourcing and research with human subjects}
    \item[] Question: For crowdsourcing experiments and research with human subjects, does the paper include the full text of instructions given to participants and screenshots, if applicable, as well as details about compensation (if any)? 
    \item[] Answer: \answerNA{} 
    \item[] Justification:  The paper does not involve crowdsourcing experiments or research with human subjects.

\item {\bf Institutional review board (IRB) approvals or equivalent for research with human subjects}
    \item[] Question: Does the paper describe potential risks incurred by study participants, whether such risks were disclosed to the subjects, and whether Institutional Review Board (IRB) approvals (or an equivalent approval/review based on the requirements of your country or institution) were obtained?
    \item[] Answer: \answerNA{} 
    \item[] Justification: The paper does not involve human-subject studies and therefore does not require IRB approval.

\item {\bf Declaration of LLM usage}
    \item[] Question: Does the paper describe the usage of LLMs if it is an important, original, or non-standard component of the core methods in this research? Note that if the LLM is used only for writing, editing, or formatting purposes and does \emph{not} impact the core methodology, scientific rigor, or originality of the research, declaration is not required.
    \item[] Answer: \answerYes{} 
    \item[] Justification: LLM is used for writing code and polishing the paper.

\end{enumerate}